\documentclass[12pt]{article}
\usepackage{caption}
\usepackage{amsmath}
\usepackage{amssymb}
\usepackage{graphicx}
\usepackage{mathtools}
\usepackage{mathrsfs}
\usepackage{hyperref}
\hypersetup{colorlinks = true, linkcolor=blue, citecolor=blue}
\usepackage{amsmath,bm}
\usepackage{graphics}
\usepackage{tikz-cd}
\usepackage{amsthm}
\usepackage{mathabx}
\usepackage{todonotes}
\usepackage{enumerate}
\usepackage{cancel}
\usepackage{soul}
\usepackage{authblk}

\newlength{\bibitemsep}
\newlength{\bibparskip}
\let\oldthebibliography\thebibliography
\renewcommand\thebibliography[1]{%
	\oldthebibliography{#1}%
	\setlength{\parskip}{\bibitemsep}%
	\setlength{\itemsep}{\bibparskip}%
}

\newcommand\smallO{
  \mathchoice
    {{\scriptstyle\mathcal{O}}}% \displaystyle
    {{\scriptstyle\mathcal{O}}}% \textstyle
    {{\scriptscriptstyle\mathcal{O}}}% \scriptstyle
    {\scalebox{.7}{$\scriptscriptstyle\mathcal{O}$}}%\scriptscriptstyle
  }

\newtheorem{theorem}{Theorem}
\newtheorem*{theorem*}{Theorem}
\newtheorem*{ftheorem*}{Formal Theorem}
\newtheorem*{definition*}{Definition}
\newtheorem{corollary}{Corollary}
\newtheorem*{remark*}{Remark}
\newtheorem{lemma}{Lemma}

\newcommand\R{\mathbb R} %Real
\newcommand\scpr{\boldsymbol{\cdot}} %scalar product dot
\newcommand\inner[1]{\langle #1 \rangle}
\newcommand\C{\mathbb C} %Complex
\newcommand\restr[2]{{% we make the whole thing an ordinary symbol
  \left.\kern-\nulldelimiterspace % automatically resize the bar with \right
  #1 % the function
  \vphantom{\big|} % pretend it's a little taller at normal size
  \right|_{#2} % this is the delimiter
  }}

\newcommand\Z{\mathbb Z} %Integer
\newcommand\N{\mathbb N} %Natural
\newcommand\Q{\mathbb Q} %Rational

\newcommand\norm[1]{\lVert #1 \rVert} % norms

\newcommand\ra{\rightarrow}

\newcommand{\rhobar}{\bar{\rho}}

\numberwithin{equation}{section}
\numberwithin{theorem}{section}
\numberwithin{corollary}{section}
\numberwithin{lemma}{section}

\begin{document}

\title{
Next-order correction to the Dirac exchange energy \\ of the free electron gas in the thermodynamic limit \\ and generalized gradient approximations
}

\author{Thiago Carvalho Corso\thanks{Email: \url{thiago.carvalho@ma.tum.de}}
\affil{Zentrum Mathematik, Technische Universit\"at M\"unchen, Germany} and Gero Friesecke\thanks{Email: \url{gf@ma.tum.de}}}
\affil{Zentrum Mathematik, Technische Universit\"at M\"unchen, Germany}

%\date{March 18, 2023}

\renewcommand\Affilfont{\itshape\small}

\maketitle

\begin{abstract}
	We derive the next order correction to the Dirac exchange energy for the free electron gas in a box with zero boundary conditions in the thermodynamic limit. The correction is of the order of the surface area of the box, and comes from three different contributions: (i) a real-space boundary layer, %in which the density gradient is of order $1$,
	(ii) a boundary-condition-induced small shift of Fermi momentum and bulk density,
	%Fourier space boundary layer, 
	and (iii) a long-range electrostatic finite-size correction. Moreover we show that the LDA, in addition to capturing the bulk term exactly, also produces a correction of the correct order but not the correct size. GGA corrections are found to be capable of capturing the surface term exactly, provided the gradient enhancement factor satisfies a simple explicit integral constraint.
	For current GGAs such as B88 and PBE we find that the new constraint is not satisfied and the size of the surface correction is overestimated by about ten percent. The new constraint  might thus be of interest for the design of future exchange functionals. 
\end{abstract}	
	%, whereas the LDA underestimates it by a similar amount.  
	%We cast correctness of the surface term into the form of a simple exact constraint on the gradient enhancement factor for GGAs.  
	%We investigate analytically and numerically the asymptotic behaviour of a few popular density functionals for the free electron gas in a box (FEGB) in the thermodynamic limit. In this limit we provide explicit formulas for the coefficients of a two term asymptotic expansion of any semi-local density functional and for the exact exchange. In particular, we show that Generalized Gradient Approximations (GGA) produce corrections of the order of the surface area of the box $L^2 \approx N^{2/3}$, where $N$ is the number of electrons and $L$ the length scale of the box, giving a precise meaning to the generally accepted fact that GGAs represent a second order correction to the local density approximations.
	
% Keywords: density functional theory, generalized gradient approximation, exchange energy, electron gas, thermodynamic limit 

%\tableofcontents

\section{Introduction}

We derive the next order correction to the Dirac exchange energy for the free electron gas in a box with zero boundary conditions in the thermodynamic limit. Like Dirac exchange, the correction is of significant interest for density functional theory (DFT). In particular, it yields a novel exact constraint on generalized gradient approximations (GGAs).

 \textit{Motivation from DFT and main result.}
 Original Kohn-Sham DFT \cite{KohnSham1965} is based on the paradigm that exchange and correlation effects in real many-electron systems in solid-state and molecular physics are reasonably similar to those in the three-dimensional uniform electron gas (UEG). This paradigm is implemented in the form of the celebrated local density approximation (LDA): one applies an analytical parametrization of the exchange-correlation energy of the UEG as a function of the density\footnote{obtained from high- and low-density asymptotics as well as numerical quantum Monte Carlo results} locally to the non-uniform density of the real system, thereby obtaining an approximation of the system's exact exchange-correlation energy. Thus the LDA corresponds to the following approximate model for the exchange-correlation energy density at a point $r$:  $e_{xc}(r)=f(\rho(r))$,  where $f(\rhobar)$ is the exchange-correlation energy density of the UEG with density $\rhobar$ and $\rho(r)$ is the density at $r$ of the inhomogeneous system. 

A milestone in elevating DFT from a useful first approximation to a widely applicable quantitative method was the development of generalized gradient approximations (GGAs) \cite{Becke1988, PBE1996}. These model the effect of density inhomogeneities by allowing the exchange-correlation energy density at a point $r$ to depend not just on the local density $\rho(r)$ but also the local density gradient, 
$e_{xc}(r)=f(\rho(r),|\nabla\rho(r)|)$. Deceptively simple but highly judicious semi-empirical choices for $f$ such as those by Becke \cite{Becke1988} and Perdew, Burke, and Ernzerhof \cite{PBE1996}  improved the accuracy of total DFT energies for small atoms and molecules from 1 eV to about 0.2 eV. 

The dominating part of the exchange-correlation energy for real molecular and solid-state systems consists of the exchange energy, on which we focus in this paper. The LDA exchange energy density is given by Dirac exchange, i.e. the exchange energy density of the \textit{free electron gas with density $\rhobar$} (which is known explicitly, see below) applied to the local density. But what about gradient corrections? Successful GGA models for exchange, such  as the widely used Becke-88 \cite{Becke1988} and PBE \cite{PBE1996} functionals, are  not obtained from uniform electron gas theory alone. Instead one makes a  low-dimensional (1- respectively 2-parameter) semi-empirical nonlinear ansatz for the gradient dependence and fits it to ab initio computations for noble gas atoms (Becke) respectively the lowest-order term in the gradient expansion of the weakly inhomogeneous electron gas\footnote{in which the electron gas is subjected to a weak and slowly-varying external potential and its energy density is expanded in increasingly higher-order density derivatives} and the Lieb-Oxford inequality (PBE). We emphasize that the gradient expansion which partially informs PBE  requires small density gradients and provides no information on the regime typical for real systems where the density gradient is of order $1$ (in atomic units). %Further difficulties with this expansion are discussed in \cite{B88, toulouse}.
%In keeping with the ``electron gas paradigm'',

It seems to us that a very  natural reference system which provides just such a regime is the free electron gas in a box which underlies the celebrated Dirac exchange, but with an important difference: one replaces periodic boundary conditions by \textit{zero boundary conditions}. While the former yield a uniform density, the latter yield density gradients of order $1$ near the boundary, see Figure~\ref{F:bdrylayer}. This is the system studied in the present paper. Note that its exchange energy is captured correctly by the exact Hohenberg-Kohn functional since the system is a high-density limit of the interacting electron gas in a box (see eq.~\eqref{highdens} below).   
The free electron gas with zero boundary conditions can be solved exactly for $N$ electrons, just as in the periodic case. One can then study its asymptotic behaviour in the thermodynamic limit where $N$ and the sidelength $L$ of the box tend to infinity with the number of electrons per unit volume, $N/L^3=\rhobar$, remaining constant.  

\begin{figure}[http!]
    \centering
    \includegraphics[width=0.4\textwidth]{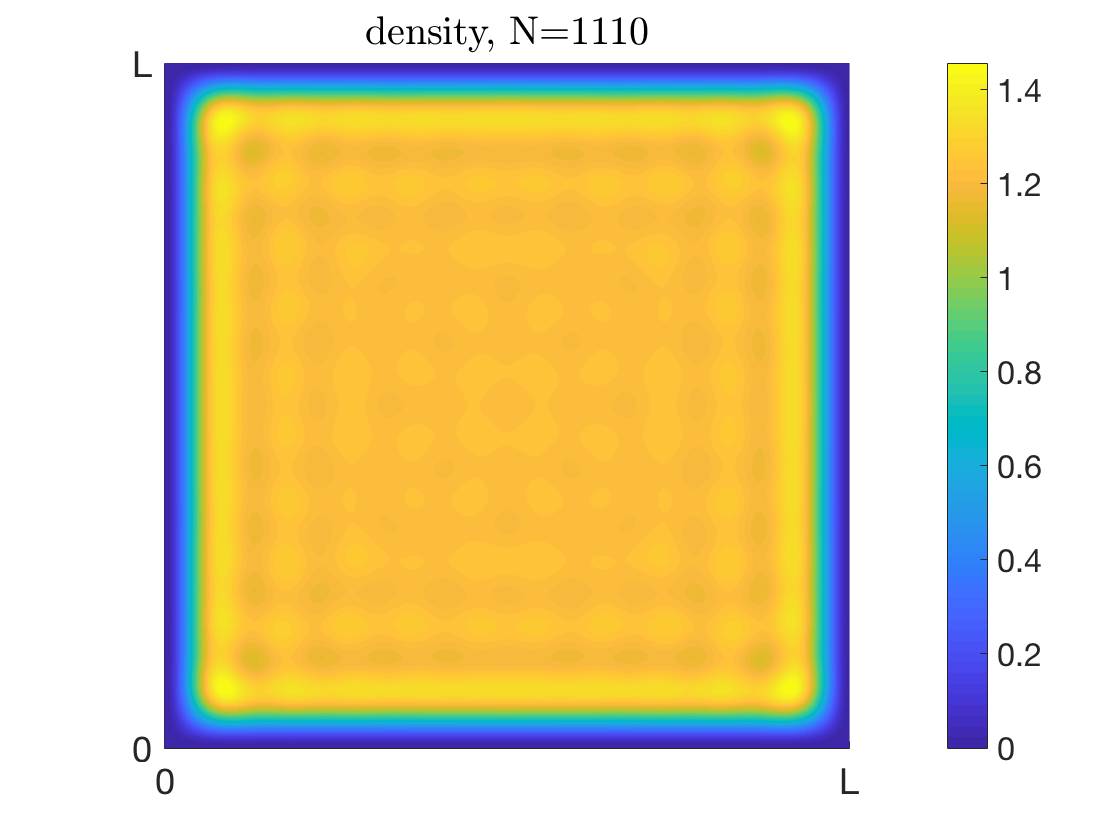} \;\;\;\; 
 \includegraphics[width=0.4\textwidth]{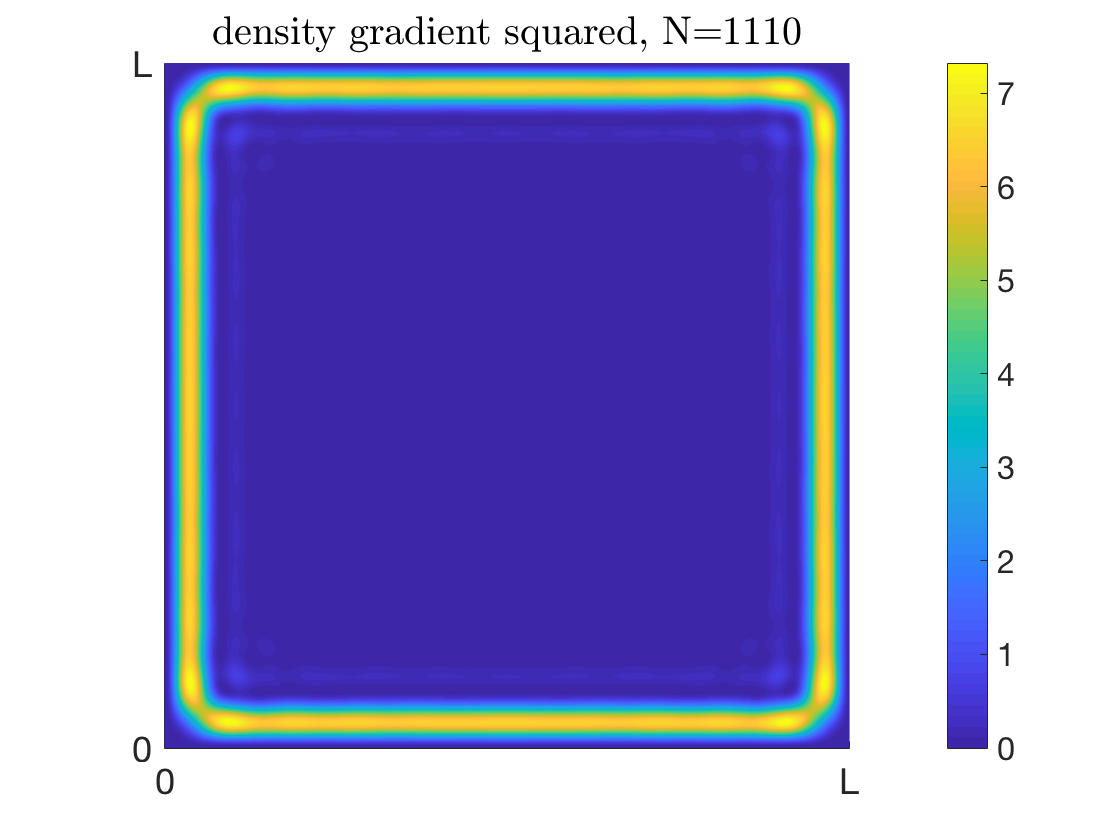}
    \caption[width=0.9\textwidth]{Density (left) and density gradient squared (right) of the free electron gas with 1110 electrons in a three-dimensional box with zero boundary conditions. The picture shows a two-dimensional cross-section through the center of the box, and the number of electrons per unit volume was normalized to $1$. 
    The density gradients of order $1$ near the boundary persist in the thermodynamic limit.
    %Note the density gradients of order $1$ near the boundary, which persist in the thermodynamic limit.
    }
    \label{F:bdrylayer}
\end{figure}

By careful asymptotic analysis, we are able to determine not just the bulk contribution to the exchange energy, which is just the familiar Dirac exchange regardless of the imposed boundary conditions (as has been shown previously  \cite{Friesecke1997}), but also the next-order (surface) contribution to the exchange energy, see Theorem \ref{mainthm} below.  The next-order term is to our knowledge new and captures the inhomogeneous boundary layer depicted in Figure \ref{F:bdrylayer}. It also captures two additional effects: a boundary-condition-induced small shift of Fermi momentum and bulk density,  and a long-range electrostatic finite-size correction which would also be present for periodic boundary conditions (i.e. unform density). Our asymptotic methods also yield the next-order (surface) contribution to the exchange energy for GGA exchange functionals with general $f$. Requiring these contributions to match yields a novel exact constraint on GGAs (see eq.~\eqref{exactconstr} below).

%Not the same as the 'small gradient expansion'; instead, a very useful reference system for obtaining new  rigorous insight into fully inhomogeneous systems, and into GGAs in DFT.

%- our findings: analyze exact (many-body quantum) $E_x$ for el.gas with Dirichlet bc's; get useful exact constraint on GGAs not satisfied by current approx's 

% To the best of our knowledge, our paper yields the first rigorous  mathematical insights into of any aspect of GGAs.

%- Criticism of GGAs: ansatz somewhat arbitrary; indeed nowadays many other approx's in use. Criticism of small-gradient expansion (well known): real systems have regimes where density gradient is O(1)

%- present work: Look at free el. gas in a box, which is just what underlies the famous Dirac exchange, but with an important difference: don't use PBC's ($\to$ homogeneous density) but Dirichlet boundary conditions ($\to$ naturally contains O(1) density gradients near the boundary). See Figure ?. 

%This system can be solved exactly for finite $N$, and by careful asymptotic analysis it is possible to determine not just the bulk contribution, which is classical, but also the surface contribution to the exchange energy, which is new. Not the same as the 'small gradient expansion'; instead, a very useful reference system for obtaining new  rigorous insight into fully inhomogeneous systems, and into GGAs in DFT.

%- our findings: analyze exact (many-body quantum) $E_x$ for el.gas with Dirichlet bc's; get useful exact constraint on GGAs not satisfied by current approx's 

\textit{Main result in more detail.} The free electron gas in a box consists of $N$ electrons moving freely in a three-dimensional box $Q_L=[0,L]^3$ of sidelength $L$ and volume $V=L^3$ in the thermodynamic limit $N\to\infty$, $V\to\infty$, with the number of electrons per unit volume, $\rhobar=N/V$, remaining constant. Mathematically, ground states of the finite system are defined as minimizers of kinetic energy
\begin{equation} \label{kinen}
   T[\Psi] =  \sum_{s_1,...,s_N\in\Z_2} \int_{Q_L^N} \frac12 \sum_{i=1}^N|\nabla_{r_i}\Psi(r_1,s_1,...,r_N,s_N)|^2 dr_1 ... dr_N
\end{equation}
over square-integrable $N$-electron wave functions $\Psi$ with finite kinetic energy (i.e., functions in the Sobolev space $H^1((Q_L\times\Z_2)^N;\C)$) subject to the following constraints: normalization, $||\Psi||_{L^2}=1$; antisymmetry,  
$\Psi(...,r_i,s_i,...,r_j,s_j,...) = - \Psi(..,r_j,s_j,...,r_i,s_i,...)$ for $i\neq j$ (where $(r_i,s_i)\in Q_L\times\Z_2$ are space-spin coordinates for the i$^{th}$ electron); and one of the boundary conditions
\begin{align}
 \Psi(r_1,s_1,...,r_N,s_N)  = \Psi(r_1',s_1,...,r_N',s_N) \;
 \mbox{ if } r'\! -r\in L\Z^{3N} \quad \quad & \mbox{(periodic case) } \label{periodicBCs}\\
  \Psi(r_1,s_1,...,r_N,s_N) = 0 \; \mbox{ if any }r_i\in\partial[0,L]^3  \quad\quad & \mbox{(Dirichlet case) } \label{dirichletBCs} \\
 \nabla_{r_i}\Psi(r_1,s_1,...,r_N,s_N) \cdot \nu(r_i) = 0 \; \mbox{ if any }r_i\in\partial[0,L]^3 \quad \quad & \mbox{(Neumann case). } \label{neumannBCs}
\end{align}
Here $\nu(r_i)$ denotes the outward unit normal to $\partial[0,L]^3$ at $r_i$. (Of course, in the Neumann case no boundary conditions are imposed on the admissible functions; instead,  minimizers then automatically satisfy Neumann conditions.)

For ground states of non-interacting systems, such as the one above, the exchange energy is defined as the \textit{difference between the quantum-mechanical electron-electron interaction energy and the mean field energy}, 
\begin{equation} \label{exch0}
    E_x[\Psi]  = V_{ee}[\Psi] - \frac12 \int_{Q_L^2} \frac{\rho(r)\rho(r')}{|r-r'|} dr\, dr'
\end{equation}
with 
\begin{equation} \label{Vee}
   V_{ee}[\Psi] = \sum_{s_1,...,s_N\in\Z^2} \int_{Q_L^N} \sum_{1\le i<j\le N} \frac{1}{|r_i-r_j|} |\Psi(r_1,s_1,...,r_N,s_N)|^2 dr_1...dr_N
\end{equation}
(interaction energy) and
\begin{equation} \label{rho}
   \rho(r) = N \sum_{s_1,...,s_N\in\Z^2}
   \int_{Q_L^{N-1}}  |\Psi(r,s_1,r_2,s_2,...,r_N,s_N)|^2 dr_2...dr_N
\end{equation}
(single-particle density of the system). For explicit expressions of the exchange energy in terms of the single-particle orbitals (Laplace eigenfunctions in the box) see Section~\ref{sec:electrongasformulas}.  

We note that the exchange energy  \eqref{exch0} is of fundamental interest as a correction term appearing in the asymptotic expansion of the interacting ground state energy in the high-density limit. More precisely, denoting the interacting ground state energy, i.e. that of $T+V_{ee}$, in $Q_L$ by $E^{\rm T + V_{ee}}_L$ and considering the re-scaling
$$
  \Psi^\gamma(r_1,s_1,...,r_N,s_N) = \gamma^{3N/2}\Psi(\gamma r_1,s_1,...,\gamma r_N,s_N) \;\;\;\; (\gamma>0),
$$
an elementary change of variables shows that 
$$
        E^{T+V_{ee}}_{L/\gamma} = \gamma^2 E^{T+\gamma^{-1}V_{ee}}_L.
$$
Hence in the high-density limit $\gamma\to\infty$, the interaction is a small perturbation and thus by standard perturbation theory, provided the ground state $\Psi$ of $T$ in $Q_L$ is non-degenerate,
\begin{equation} \label{highdens}
   \tfrac{1}{\gamma^2} E^{T+V_{ee}}_{L/\gamma} = T[\Psi] + 
   \tfrac{1}{\gamma} \Bigl( J[\rho] + E_x[\Psi]\Bigr) + O\bigl(\tfrac{1}{\gamma^2}\bigr) \;\;\; \mbox{ as }\gamma\to\infty
\end{equation}
where $J[\rho]$ denotes the mean-field energy (second term in \eqref{exch0}). See the review \cite{GillUEG2016} for more information on the interacting electron gas and the recent review \cite{SCElimit2022} for information about the opposite (low-density) limit in which electrons become strictly correlated. 

We are not just interested in the exact exchange energy functional \eqref{exch0}, but also want to compare it to two important types of simpler functionals defined only in terms of the single-particle density:

\begin{itemize}
\item The Local Density Approximation (LDA) \cite{KohnSham1965}: 
\begin{align} E_x^{\rm LDA}[\rho] = \int_{\R^3} e_x^{\rm LDA}(\rho(r)) \mathrm{d}r \label{LDA} \end{align}
where the exchange energy density per unit volume is given by the Dirac-Bloch formula \cite{Bloch1929,dirac_1930} $e^{\rm LDA}_x(\rho) = -c_x \rho^{4/3}$ with $c_x = \frac{3}{4}(\frac{3}{\pi})^{\frac{1}{3}}$.
\item The GGA functionals \cite{Becke1988,PBE1996,PBEsol2008}:
\begin{align} E_x^{\rm GGA}[\rho] = 
E^{\rm LDA}_x[\rho] +\underbrace{\int_{\R^3} g^{\rm GGA}(\rho(r),|\nabla \rho(r)|) \mathrm{d}r}_{\coloneqq \Delta E_x^{\rm GGA}[\rho]} \label{gga} \end{align}
with the assumption that $g^{\rm GGA}(\bar{\rho},0) = 0$ for all $\bar{\rho}\geq 0$ (i.e. the functional reduces to the LDA for the homogeneous density). 
\end{itemize}
In the physics literature \cite{Becke1988,PBE1996,PBEsol2008}, GGAs are commonly expressed in terms of the density and the dimensionless gradient $s=|\nabla\rho|/\rho^{4/3}$. This has the advantage that, by a scaling argument, one arrives at the simpler ansatz
\begin{equation} \label{physlit}
   g^{\rm GGA}(\rho,|\nabla\rho|) = e_x^{\rm LDA}(\rho)G(s),
\end{equation}
with different GGAs differing only by the choice of $G$.\footnote{I.e., by the ``gradient enhancement factor'' $F = 1 + G$ of the overall integrand $f(\rho,|\nabla\rho|)=e_x(\rho)+g^{\rm GGA}(\rho,|\nabla\rho|)=e_x(\rho)F(s)$} The reason we prefer to work with the density and the density gradient instead is because $s(r) \ra \infty$ as $r$ approaches the boundary (for Dirichlet boundary conditions) while $\nabla \rho(r)$ remains bounded, making the mathematical analysis simpler. Our assumptions on $g^{\rm GGA}$ required in Theorem \ref{mainthm} below are as follows:
\begin{equation} \label{ggaassump}
 g^{\rm GGA}\in C^0([0,\infty)^2)\cap C^1((0,\infty)\times[0,\infty)) \mbox{ with }g^{\rm GGA}(\rhobar,0) = 0\mbox{ for all }\rhobar. 
\end{equation}
These are satisfied for typical GGAs of form \eqref{physlit} such as those in \cite{Becke1988,PBE1996,PBEsol2008} (see Appendix~\ref{sec:GGAcheck} for a proof).

With the functionals \eqref{exch0}, \eqref{LDA}, \eqref{gga} in mind, the main result of this paper can be stated as follows.
\begin{theorem}[Asymptotic expansion of exchange functionals] \label{mainthm} Let $N \in \N,L>0$, and let $\Omega \subset\R^3$ be a rectangular box. Let $\Psi_{N,L}$ be any determinantal ground state wave function of the free $N$-electron gas in $\Omega_L = \{ x \in \R^3 : x/L \in \Omega \}$ under either Dirichlet, Neumann, or periodic boundary conditions, and let $\rho_{N,L}$ denote the associated single-particle density. Moreover, assume that the GGA functional \eqref{gga} satisfies \eqref{ggaassump}. Then in the thermodynamic limit, i.e., for $N,L \ra \infty$ and $\bar{\rho} =N/(|\Omega|L^3) = constant$, one has:
\begin{itemize}
\item Under periodic boundary conditions:
\begin{align*} E_x[\Psi^{\textnormal{Per}}_{N,L}] &= -c_x \bar{\rho}^{4/3} |\Omega|L^3+c_{x,2}^{\textnormal{Per}}\bar{\rho}|\partial \Omega|L^2 + \mathcal{O}(L^{\frac{45}{23}+\epsilon}) \\ E^{\rm LDA}_x[\rho^{\textnormal{Per}}_{N,L}] &= -c_x \bar{\rho}^{4/3} |\Omega|L^3 + \mathcal{O}(L^{\frac{34}{23}+\epsilon}) \\ \Delta E_x^{\rm GGA}[\rho^{\textnormal{Per}}_{N,L}]&= \mathcal{O}(L^{\frac{34}{23}+\epsilon})  \end{align*}
\item Under Dirichlet boundary condtions:
\begin{align*} E_x[\Psi^{\textnormal{Dir}}_{N,L}] &= -c_x \bar{\rho}^{\frac{4}{3}} |\Omega|L^3 -  c_{x,2}^{\textnormal{Dir}}\bar{\rho}|\partial \Omega|L^2 + \mathcal{O}(L^{\frac{45}{23}+\epsilon}) \\ E_x^{\rm LDA}[\rho^{\textnormal{Dir}}_{N,L}] &= -c_x \bar{\rho}^{\frac{4}{3}} |\Omega|L^3 - c_{\rm LDA}^{\textnormal{Dir}} \bar{\rho}|\partial \Omega|L^2 + \smallO(L^2) \\ \Delta E_x^{\rm GGA}[\rho^{\textnormal{Dir}}_{N,L}] &= c_{\rm GGA}^{\textnormal{Dir}}(\bar{\rho})|\partial \Omega|L^2 + \smallO(L^2) \end{align*}
\item Under Neumann boundary conditions:
\begin{align*} E_x[\Psi^{\textnormal{Neu}}_{N,L}] &= -c_x \bar{\rho}^{\frac{4}{3}}|\Omega|L^3 -c_{x,2}^{\textnormal{Neu}} \bar{\rho}|\partial \Omega|L^2 + \mathcal{O}(L^{\frac{45}{23}+\epsilon}) \\ 
E_x^{\rm LDA}[\rho^{\textnormal{Neu}}_{N,L}] &= -c_x \bar{\rho}^{\frac{4}{3}} |\Omega|L^3 - c_{\rm LDA}^{\textnormal{Neu}} \bar{\rho}|\partial \Omega|L^2 + \smallO(L^2)  \\
 \Delta E_x ^{\rm GGA}[\rho^{\textnormal{Neu}}_{N,L}]&=  c_{\rm GGA}^{\textnormal{Neu}}(\bar{\rho}) |\partial \Omega|L^2+ \smallO(L^2) \end{align*}
\end{itemize}
where $|\Omega|$ and $|\partial \Omega|$ denotes the volume and surface area of the domain $\Omega$, $h$ is the explicit function $h(t)= 3 (\sin t - t \cos t)/t^3$, $p_F = (3\pi^2 \bar{\rho})^{1/3}$ (Fermi momentum), and the constants are given by
\begin{align*} &c_{x,2}^{\textnormal{Per}} = \frac{1}{8}, \quad c_{x,2}^{\textnormal{Dir}} = \frac{1-\log 2}{4} \approx 0.0767, \quad c_{\rm LDA}^{\textnormal{Dir}} = \frac{3}{8\pi}\int_0^\infty (1-h(t))^{\frac{4}{3}}-1 \mathrm{d}t + \frac{3}{8} \approx 0.0673, \\
 &c_{x,2}^{\textnormal{Neu}} =  \frac{ 3 \log 2 -2}{4} \approx 0.0199 , \quad c_{\rm LDA}^{\textnormal{Neu}} = \frac{3}{8\pi}\int_0^\infty (1+h(t))^{\frac{4}{3}}-1 \mathrm{d}t -\frac{3}{8} \approx 0.0430 ,  \\
  & c_{\rm GGA}^{\textnormal{Dir}}(\bar{\rho}) = \frac{1}{2p_F}\int_0^\infty g^{\rm GGA}\biggr(\bar{\rho}(1-h(t)),2\bar{\rho}p_F|\dot{h}(t)|)\biggl) \mathrm{d}t , \\
  & c_{\rm GGA}^{\textnormal{Neu}}(\bar{\rho}) = \frac{1}{2p_F}\int_0^\infty g^{\rm GGA}\biggr(\bar{\rho}(1+h(t)),2\bar{\rho}p_F|\dot{h}(t)|\biggl) \mathrm{d}t. \end{align*}
\end{theorem}

This result extends that of a previous work by one of the authors \cite{Friesecke1997} as we determine not just the leading but also the next-order terms (of order $L^2$) and include the GGA functionals. Also to further illustrate the role of the boundary conditions we have included the Neumann case. 

An immediate corollary of Theorem \ref{mainthm} is the following simple exact constraint on GGAs. The next-order correction to Dirac exchange for the free electron gas with zero boundary conditions is captured exactly, i.e. 
\begin{align} \frac{E_x[\Psi_{N,L}] - E_x^{\rm GGA}[\rho_{N,L}]}{L^2} \ra 0 \quad \mbox{ as } N,L\ra \infty \mbox{ with } \bar{\rho} = \frac{N}{L^3} = constant \label{constraint} \end{align}
for all values of the average density $\rhobar$, if and only if the gradient enhancement factor $F(s)=1+G(s)$ defined by  \eqref{physlit} satisfies
\begin{align}
    \frac{3 }{8\pi}\int_0^\infty (1-h(t))^{\frac{4}{3}}G\biggr(2(3\pi^2)^{\frac{1}{3}}\frac{|\dot{h}(t)|}{(1-h(t))^{\frac{4}{3}}}\biggr) \mathrm{d} t = c_{x,2}^{\textnormal{Dir}}-c^{\textnormal{Dir}}_{\rm LDA} \label{exactconstr}
\end{align} 
where the constants and the function $h$ are those from Theorem \ref{mainthm}. In contrast with previous exact conditions on $G$ which refer to small-$s$ asymptotics \cite{PBE1996} (for the weakly inhomogeneous electron gas) respectively large-$s$ asymptotics \cite{Becke1988} (for atomic densities), the above  condition is an \textit{integral constraint} which sees the whole profile of $G$. Note that as $t$ varies from $0$ to $\infty$, the argument $s$ of $G$ (which corresponds to the reduced density gradient of the Dirichlet free electron gas along a ray moving from the boundary in perpendicular direction into the interior, see below) traces out all possible $s$ values from $\infty$ to $0$. The extent to which current GGAs fail to satisfy \eqref{exactconstr} is discussed in Section \ref{sec:numerics}. 

Finally we remark that, due to \eqref{highdens}, our results correspond to first taking a high-density limit and then a thermodynamic limit of the interacting $N$-electron gas in a box. 
The opposite order of limits was studied in \cite{graf1994correlation} (note that when first taking a thermodynamic limit the surface correction terms disappear).  
In the periodic case a coupled high-density/thermodynamic limit in which very different corrections (of RPA type) appear is studied in \cite{Christiansen2023}.

\textit{Strategy of the proof.} We follow the overall strategy introduced in \cite{Friesecke1997} of deriving an accurate continuum approximation to the ground state density matrix (see Theorem~\ref{continuumapproxthm}
below, or Theorems 4.1 and 4.2 in \cite{Friesecke1997}) and analyzing the ensuing interior and boundary contributions to the exchange energy. 
While the continuum approximation is the same already introduced in \cite{Friesecke1997}, the main advance, and most involved part of our work, is an improved error estimate 
%{in your version you forgot to say where the reader finds this estimate}
(see Theorem \ref{continuumapproxthm}) which shows that it is accurate enough to infer the next-order contributions to the exchange functionals which are of the order of the surface area of the box. This is achieved by leveraging, on top of Fourier analysis techniques \cite{sogge2017fourier} as already used in \cite{Friesecke1997}, the theory of exponential sums \cite{graham_kolesnik_1991,HeathBrown2017}. The main step is the proof of the following technical lemma which relies on recent work by Heath-Brown \cite{HeathBrown2017}. 

\begin{lemma}\label{central} Let $\alpha \in \N_0^3$ and $D \in \R^{3\times 3}$ be a positive diagonal matrix. Then there exists $c= c(\alpha,D)>0$ such that
\begin{align} \biggr | \sum_{k\in \Z^3\cap B^D_R} (ik)^\alpha e^{i k\scpr z} - \int_{B^D_R} (i k)^\alpha e^{i k\scpr z} \mathrm{d}k \biggl| \leq c(1+R^{|\alpha|+\frac{34}{23}+\epsilon}) , \label{centralest} \end{align}
for all $R>0$ and $z$ with $|z|_{max}\coloneqq \max_{j\leq 3} \{|z_j|\}\leq \pi$, where $\N_0 = \N\cup \{0\}$ and $B^D_R\coloneqq \{k\in \R^3 : |D^{-1} k|\leq R\}$.
\end{lemma}

The exponent $\frac{34}{23}$ may seem peculiar at first, and we do not claim it to be optimal, but the main point is that it improves over the $\frac{3}{2}$ exponent obtained in \cite{Friesecke1997}. This improvement is necessary for rigorous derivation of the asymptotic terms of the order of magnitude of the surface area of the box. This can be quickly seen by integrating the square of an error proportional to $L^{-\frac{3}{2}}$ (like in Theorem 4.1 and 4.2 of \cite{Friesecke1997}) against the Coulomb potential in the double box $[0,L]^6$, which yields an error proportional to $L^2$ and is therefore not enough for our purpose.

Estimates of this kind were originally motivated by analytic number theory. In particular, by setting $z=0$ and $\alpha=0$ one recognizes the famous lattice point counting problem in $\R^3$ (also known as the sphere problem) \cite{Vinogradov1963,chen1963improvement,Chamizo1995,HeathBrown1999}. The remarkable difference here is that the estimate holds uniformly in $z$, even though the integrand oscillates, for typical $k\in B_R^D$, on the length scale of the lattice. 

\begin{figure}[ht!]
    \centering
    \includegraphics[width=0.7\textwidth]{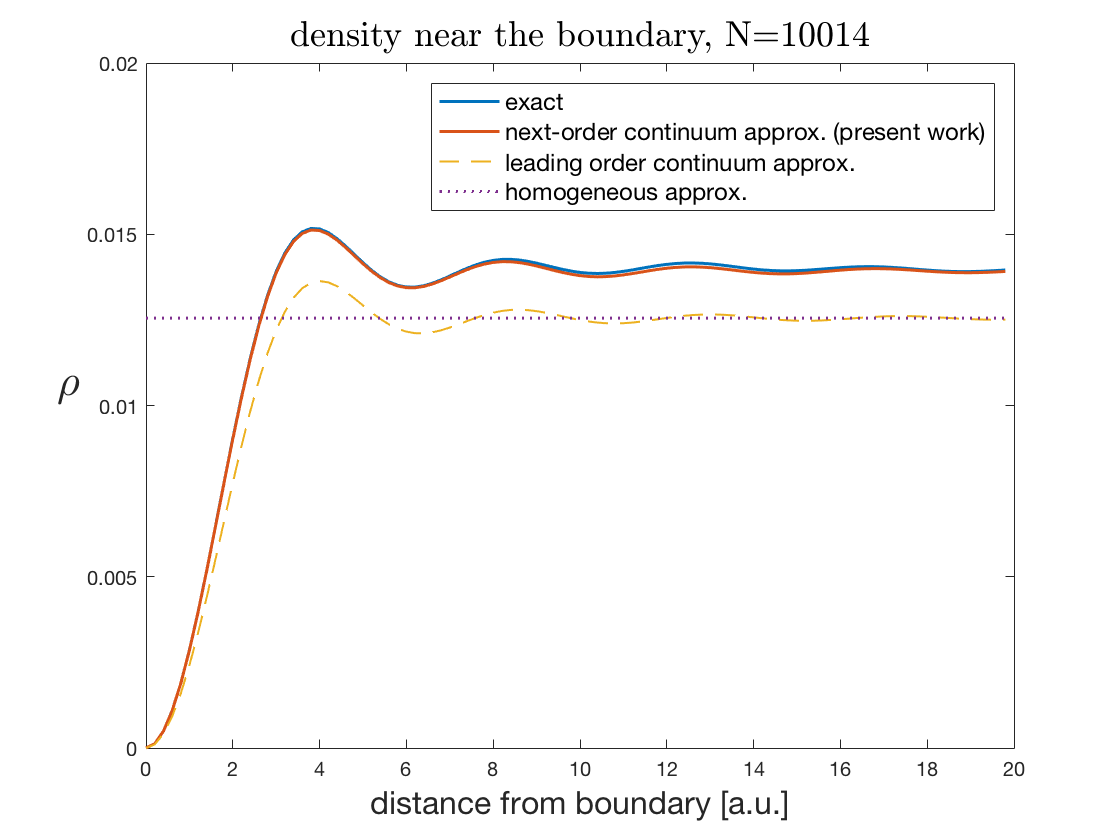}
    \caption{Boundary layer of the free electron gas with 10014 electrons in a three-dimensional box with zero boundary conditions. The exact density is given by \eqref{densitydir}, the leading order continuum approximation by its thermodynamic limit \eqref{layerapprox3}, and the next-order continuum approximation by eqs.~\eqref{layerapprox}--\eqref{layerapprox2}. The number of electrons per unit volume was normalized to the conduction electron density of copper.}
    \label{F:layerapprox}
\end{figure}

The continuum approximation of the density matrix which we justify with the help of the above lemma (see Theorem \ref{continuumapproxthm}) entails, in particular, the following accurate approximation to the boundary layer for zero boundary conditions and the box $[0,L]^3$: 
\begin{equation}  \label{layerapprox}
  \rho_{N,L}(r_0+r') = \rhobar\Bigl( \frac{p_{N,L}}{p_F} \Bigr)^3\bigl( 1 - h(2p_{N,L} |r'|)\bigr) + O(L^{-\frac{35}{23} + \varepsilon})
\end{equation}
whenever $r_0$ belongs to the boundary of the box, its distance from the edges is of order $L$, and $r'$ points in normal direction to $\partial[0,L]^3$ into the interior. Here $p_{N,L}$ is the Fermi momentum of the finite system, which is found (see Lemma \ref{latticeeffect}) to differ from its thermodynamic limit $p_F=(3\pi^2\rhobar)^{1/3}$ by an order $L^{-1}$ shift, 
\begin{equation} \label{layerapprox2}
   p_{N,L} = p_F + \frac{3\pi}{4}L^{-1} + O(L^{\frac{35}{23}+\epsilon}).
\end{equation}
An important and rather subtle consequence of \eqref{layerapprox}--\eqref{layerapprox2} is the following. The -- at first sight reasonable -- leading order approximation (thermodynamic limit) of the boundary layer profile in normal direction,
\begin{equation} \label{layerapprox3}
\lim_{\substack{ N, \, L \to\infty \\ N/L^3=\rhobar}} \rho_{N,L}(r_0+r') = \rhobar\bigr(1-h(2p_F|r'|)\bigr), 
\end{equation}
while enough to resolve the surface contribution for GGA exchange\footnote{note that the integrand in \eqref{exactconstr} is precisely the GGA correction to the exchange energy density for this profile after making the change of variables $t=2p_F|r'|$ and setting $\rhobar$=1}, \textit{yields a surface contribution of the wrong size for exact exchange}, missing e.g. the $\sim L^{-1}$ rearrangement of the bulk density implied by \eqref{layerapprox}--\eqref{layerapprox2}. The much better accuracy of the next-order approximation \eqref{layerapprox}--\eqref{layerapprox2} can be seen in Figure \ref{F:layerapprox}.

\textit{Structure of the paper.} We start with a small subsection to introduce the notation used throughout the paper. In Section~\ref{sec:electrongasformulas} we begin by recalling some basic facts about the ground state of the free electron gas in the box under different boundary conditions. In Section~\ref{sec:fermimomentum} we discuss the control of open shell effects and the Fermi momentum asymptotics in the thermodynamic limit. Section~\ref{sec:continuumapprox} contains the proof of Lemma \ref{central} and the derivation of the continuum approximation of the density matrix. In Section~\ref{sec:exchangethm} we present the proof of Theorem \ref{mainthm} by splitting it into treating semi-local functionals (Theorem \ref{generalasymp}) and exact exchange (Theorem \ref{exactexchangethm}). Section~\ref{sec:kinetic} briefly discusses the asymptotics of the kinetic energy, which can easily be extracted with our methods. Section~\ref{sec:numerics} compares the asymptotic behaviour of different exchange functionals  
(exact exchange, LDA, B88, PBE, PBEsol) when applied to the free electron gas in a box with zero boundary conditions and up to 30~000 electrons. We find good agreement between asymptotics and numerics. Physics-minded readers may want to skip Sections~\ref{sec:electrongasformulas}--\ref{sec:kinetic} and move forward directly to Section~\ref{sec:numerics}.

\subsection{Notation}

The following notation will be used throughout the text.
\begin{itemize} \item We use the standard big-O and small-o notation: for functions $f:(0,\infty) \ra \R$ and $g:\R\ra (0,\infty)$, we say that $f = \mathcal{O}(g)$ respectively $f = \smallO(g)$ if
\begin{align*} \limsup_{L\ra\infty} \frac{|f(L)|}{g(L)} <\infty \quad \mbox{ respectively } \quad \limsup_{L\ra \infty} \frac{|f(L)|}{g(L)} = 0.  \end{align*}
Moreover, for functions $f,g:\R \mapsto \R$, we say that $f(L) \lesssim g(L)$ or $f(L) \sim g(L)$ to indicate, respectively, the existence of a constant $C>0$ which does not depend on $L$ such that 
\begin{align*} |f(L)| \leq C |g(L)| \quad \mbox{ or } \quad C^{-1} |f(L)| \leq |g(L)| \leq C |f(L)| \end{align*}
for all sufficiently large values of $L$. Sometimes we will also use the notation $\lesssim_{\epsilon}$ to indicate dependence of the implicit constant on an additional parameter ($\epsilon$ in this case).
\item Throughout the text, $D \in \R_+^{3\times 3}$ will always denote a  diagonal matrix with entries $d_1,d_2,d_3 >0$, and $|D| \coloneqq \det D$ stands for its determinant. The balls of radius $R$ and $D$-radius $R$ are denoted by
\begin{align*} B_R \coloneqq \{ r\in \R^3 : |r|\leq R\} \quad \mbox{ and } \quad B^D_R \coloneqq \{ r\in \R^3 : |D^{-1} r|\leq R\}. \end{align*}
The cubic box, the $D$-rectangular box, and their re-scaled versions are denoted by
\begin{align*} Q \coloneqq [0,1]^3, \;\;\;  & Q^D  \coloneqq [0,d_1]\times [0,d_2]\times [0,d_3] , \\
Q_L \coloneqq [0,L]^3, \;\;\; & Q_L^D \coloneqq [0,L\, d_1]\times [0, L \, d_2]\times [0,L \, d_3]. \end{align*}
\item  For the Fourier transform of a function $f:\R^n \ra \C$, we use the normalization convention
\begin{align} \widehat{f}(k) = \int_{\R^n} f(r) e^{-i k\scpr r} \mathrm{d}r \label{fourierconvention} \end{align}
where $k\scpr r \coloneqq \sum_{j=1}^n k_j r_j$ is the standard Euclidean scalar product. We denote the inverse Fourier transform of $f$ by $\widecheck{f}$.
\item For a set $\Omega \subset \R^n$, we use $\chi_{\Omega}$ for %denoted 
its characteristic function.  In particular, with the above convention for  the Fourier transform, in $\R^3$ one has that $\widehat{\chi}_{B_1} = \frac{4\pi}{3} h(|k|)$, where the function $h:\R \mapsto \R$ will appear many times in the sequel and is given by
\begin{align} h(t) = 3\frac{\sin(t) - t \cos (t)}{t^3}. \label{hdef} \end{align}
For an elementary derivation of this formula see e.g. \cite{Friesecke1997} Lemma 6.1.
\item The group generated by reflections at coordinates hyperplanes of $\R^3$ is denoted by $G$ and its elements by $\sigma$, i.e.,
\begin{align} G = \{\sigma \in \R^{3\times 3} : \sigma \mbox{ diagonal and } \sigma_{jj} = \pm 1 \mbox{ for any } j=1,2,3\}. \label{reflectiondef}  \end{align}
\item The projection on the $i^{th}$ coordinate hyperplane  is denoted by $\pi_i : \R^3 \ra \R^2$, e.g. $\pi_1(r_1,r_2,r_3) = (r_2,r_3)$. Moreover, for any $z\in \R^3$, we define $|z|_{\max} \coloneqq \max_{j\leq 3} |z_j|$ .
\item The positive reals, non-negative integers and the additive group of order 2 are denoted, respectively, by  $\R_+ = (0,\infty)$, $\N_0 = \N \cup \{0\}$ and $\Z_2 = \{0,1\}$.
\item For any set in $\R^3$, we use $|\cdot|$ for either its volume, surface area or cardinality depending on whether the set has dimension 3, 2, or 0, e.g, $|B_1| = \frac{4\pi}{3}$, $|\partial B_1| = 4\pi$, and $|B_1 \cap \Z^3|$ is the number of elements in $\Z^3$ with Euclidean norm smaller than $1$.
\end{itemize}

\section{Ground state of the free electron gas : closed shell formulas}
\label{sec:electrongasformulas}
We now recall some basic facts and formulas for the ground state of the free $N$-electron gas in the box subject to Dirichlet, Periodic or Neumann boundary conditions. (The reader can find a detailed account of the eigenstates in   \cite{ReedSimonIV} and of the density matrices and exchange energy in  \cite{Friesecke1997}.)

It is well known that the Laplacian in $Q^D_L$ under any of the discussed boundary conditions (BCs) is diagonalizable in the sense that there exists an orthonormal basis of $L^2(Q^D_L)$ consisting of eigenvectors. Furthermore, the eigenvectors and eigenvalues can be labelled by
\begin{itemize} 
\item vectors $k\in \Z^3$ in case of periodic boundary conditions:
\begin{align} \phi^L_k(r) = \frac{1}{\sqrt{|D| L^3}}  \, e^{i \frac{2 \pi}{L} D^{-1} k \scpr r},\quad \lambda_k = \frac{4 \pi^2|D^{-1}k|^2}{L^2} . \label{eigenvectorsper} \end{align}
\item vectors $k \in \N^3$ in case of Dirichlet boundary condition:
\begin{align} \phi^L_k(r) = \frac{1}{\sqrt{|D| L^3}} \prod_{i=1}^3 \sqrt{2}\sin\biggr(k_i \frac{\pi}{d_i L} r_i\biggr), \quad \lambda_k = \frac{\pi^2|D^{-1} k|^2}{L^2} . \label{eigenvectors} \end{align}
\item vectors $k\in\N_0^3$ in case of Neumann boundary conditions:
\begin{align} \phi^L_k(r) = \frac{1}{\sqrt{|D| L^3}}   \prod_{\substack{i=1 \\ k_i\neq 0}}^3 \sqrt{2}\cos \biggr(k_i \frac{\pi}{d_i L} r_i\biggr), \quad \lambda_k = \frac{\pi^2|D^{-1} k|^2}{L^2} . \label{eigenvectorsneu} \end{align}\end{itemize}
As a consequence, one possible ground state for the $N$-electron gas (i.e. a normalized anti-symmetric minimizer of \eqref{kinen} under one of the BCs \eqref{periodicBCs}--\eqref{neumannBCs}) with $N$ even is given by the determinantal wave function (or Slater determinant)
\begin{align} \Psi_{N,L}(x_1,...,x_N) = \frac{1}{\sqrt{N!}} \det \begin{pmatrix} \psi_1(x_1) & \ldots & \psi_1(x_N) \\ \vdots & & \vdots \\ \psi_N(x_1) 
& \ldots & \psi_N(x_N)   \end{pmatrix} , \label{slaterdeterminant} \end{align}
where $x_\ell=(r_\ell,s_\ell) \in Q^D_L \times \Z_2$ are the space-spin variables and $\psi_i$ are the space-spin orbitals given by
\begin{align} \psi_{2i-1}(x) = \phi^L_{k_i}(r)\chi_{1}(s) , \quad \quad \psi_{2i}(x) = \phi^L_{k_i}(r) \chi_{0}(s) \quad \mbox{ for } i \in\biggr\{1,...,\frac{N}{2}\biggr\} ,\label{doublyoccupied} \end{align} 
where $\phi^L_{k_i}$ are the eigenfunctions defined in \eqref{eigenvectorsper}--\eqref{eigenvectorsneu}, and $\{k_i\}_{i\leq N/2}$ is any subset of distinct vectors in $\N^3$ (in the Dirichlet case) such that $B_{\max_i |k_i|-\epsilon}\cap \N^3 \subset \{k_i\}_{i\leq N/2}$, for all $\epsilon>0$. In fact, the collection of all such Slater determinants forms a basis for the ground state eigenspace of the free $N$-electron gas in $Q^D_L$.

Let us also introduce the Fermi radius\footnote{Note that the Fermi radius also depends on $D$. However, as $D$ will be fixed and $N,L$ will vary, we will not exhibit this dependence in our notation.}
\begin{align} R_N \coloneqq \begin{dcases} \min \{ R>0 : \frac{N}{2}\leq |B^D_{R}\cap \Z^3|\}, &\mbox{ for periodic BCs,} \\ 
 \min \{ R>0 : \frac{N}{2}\leq |B^D_{R}\cap \N^3|\}, &\mbox{ for Dirichlet BCs,}  \\
 \min \{ R>0 : \frac{N}{2}\leq |B^D_{R}\cap \N_0^3|\}, &\mbox{ for Neumann BCs.} \end{dcases} \label{fermiradiidef} \end{align} 
 In the following, for boundary-condition-dependent quantities like $R_N$ we will use   superscripts Per, Dir, and Neu.
The ground state of the free-electron gas is unique for any $N$ satisfying the closed shell condition
\begin{align} 
    \frac{N}{2} = \begin{dcases} |B^D_{R^{\textnormal{Per}}_N} \cap \Z^3|, &\mbox{for periodic BCs,} \\
    |B^D_{R^{\textnormal{Dir}}_N} \cap \N^3|, &\mbox{for Dirichlet BCs,} \\
    |B^D_{R^{\textnormal{Neu}}_N} \cap \N_0^3|, &\mbox{for Neumann BCs.}\end{dcases} \label{eq:closedshellcondition}
\end{align}
In particular, by recalling that the spinless single-particle density matrix of a state $\Psi$ is defined as
\begin{align} \gamma_{\Psi}(r,\tilde{r}) \coloneqq  N \!\!\!\!\! \sum_{s_1,...,s_N\in \Z^2}\int_{(\R^3)^{N-1}}\!\!\Psi(r,s_1,r_2,s_2,...r_N,s_N) \overline{\Psi(\tilde{r},s_1,r_2,s_2,...,r_N,s_N)}\mathrm{d}r_2...\mathrm{d}r_N , \label{densitymatrixdef} \end{align}
for any $N$ satisfying the closed shell condition \eqref{eq:closedshellcondition}, the spinless density matrix of the (unique) aforementioned ground state is given by
\begin{align}
    \gamma_{N,L}(r,\tilde{r}) = \begin{dcases} \frac{2}{|D| L^3} \sum_{k\in \Z^3\cap B^D_{R_N^{\textnormal{Per}}}} e^{i \frac{2\pi}{L} D^{-1} k \scpr (r-\tilde{r})}, &\mbox{ for periodic BCs,} \\
    \frac{1}{4|D| L^3} \sum_{\sigma \in G} \det \sigma \sum_{k\in \Z^3\cap B^D_{R_N^{\textnormal{Dir}}}} e^{i\frac{\pi}{L} D^{-1}k\scpr (r-\sigma \tilde{r})}, &\mbox{ for Dirichlet BCs,} \\
     \frac{1}{4|D|L^3} \sum_{\sigma \in G} \sum_{k \in \Z^3 \cap B^D_{R_N^{\textnormal{Neu}}}}e^{i \frac{\pi}{L}D^{-1} k\scpr (r-\sigma \tilde{r})}, &\mbox{ for Neumann BCs,} \end{dcases} \label{densitymatrix}
\end{align}
where $G$ is the reflection group defined in \eqref{reflectiondef}. Therefore, the single-particle (or one-body) density can be written as
\begin{align} \rho_{N,L}(r) = \begin{dcases}  \gamma^{\textnormal{Per}}_{N,L}(r,r) = \frac{N}{|D|L^3} , &\mbox{ for periodic BCs,} \\
\gamma^{\textnormal{Dir}}_{N,L}(r,r) =  \frac{1}{4|D|L^3} \sum_{\sigma \in G} \det \sigma \sum_{k\in \Z^3 \cap B^D_{R_{N}^{\textnormal{Dir}}}} e^{i\frac{\pi}{L}D^{-1} k \scpr (r-\sigma r)} , &\mbox{ for Dirichlet BCs,} \\
\gamma^{\textnormal{Neu}}_{N,L}(r,r) =  \frac{1}{4|D|L^3} \sum_{\sigma \in G} \sum_{k\in \Z^3 \cap B^D_{R_N^{\textnormal{Neu}}}} e^{i\frac{\pi}{L}D^{-1} k \scpr (r-\sigma r)}  , &\mbox{ for Neumann BCs,} \end{dcases} \label{densitydir}
\end{align}
and the exact exchange energy can be rewritten\footnote{The equivalence of \eqref{exch0} and \eqref{exchDM} with $\gamma_\Psi$ defined by \eqref{densitymatrixdef} is in fact valid for any Slater determinant $\Psi$ of doubly-occupied spatial orbitals. It follows from \eqref{slaterdeterminant}-\eqref{doublyoccupied} by straightforward calculation.} as
\begin{align} E_x[\Psi_{N,L}] = -\frac14\int_{Q^D_L\times Q^D_L} \frac{|\gamma_{N,L}(r,\tilde{r})|^2}{|r-\tilde{r}|} \mathrm{d}r\mathrm{d}\tilde{r} , \label{exchDM} \end{align}
with $\gamma_{N,L}$ from equation \eqref{densitymatrix}.

\begin{remark*} The above expression of density and density matrix as an average of an exponential sum over the reflection group was introduced and derived in \cite{Friesecke1997}. For the Neumann case the derivation is analogous. \end{remark*}

The above expressions have a few simple but important symmetries that we state as a lemma for further reference. (The proof is a straightforward verification.)
\begin{lemma}[Symmetries of $\rho_{N,L}$]
\label{rhosymmetries} Let $\rho_{N,L}$ be defined by equation \eqref{densitydir}, then $|\nabla \rho_{N,L}|$ and $\rho_{N,L}$ are unchanged under the following reflections:
\[  r_i \mapsto d_iL-r_i 
. \]\end{lemma}

\section{Open shell effects and Fermi momentum asymptotics in the thermodynamic limit} \label{sec:fermimomentum}

In this section we consider two important aspects of the FEG in the thermodynamic limit: (i) we justify the use of formulas \eqref{densitymatrix} for determinantal ground states with a general number of particles $N \in \N$ (Lemma \ref{openshellcontrol}), and (ii) we derive a two-term asymptotic formula for the finite-size Fermi momentum (Lemma \ref{latticeeffect}).

\subsection{Open shell effects}

Our goal now is to show that, in the thermodynamic limit, the single-particle density matrix (and its derivatives) for any determinantal ground state is pointwise close to the unique closed shell formulas in \eqref{densitymatrix}. 

To this end, let us introduce the previous and the current shells for some $N\in \N$ as
\begin{align*} 
   &N_- \coloneqq \max\{ n\in \N : n < N \mbox{ and } n/2 = |B^D_{R}\cap \N^3| \mbox{ for some } R>0\} = \max_{\epsilon>0} 2|B^D_{R_N-\epsilon}\cap \N^3|, \\  
   &N_+ \coloneqq \min\{ n \in \N : n\geq N \mbox{ and }n/2 = |B^D_R\cap \N^3| \mbox{ for some } R>0 \} = 2|B^D_{R_N}\cap \N^3|,
\end{align*}
where again $\N$ is replaced by $\N_0$ and $\Z$ for Neumann and periodic boundary conditions, respectively. Then according to \cite[Section 3]{Friesecke1997}, one can work with the unique (closed shell) ground state density matrix $\gamma_{N_-,L}$ up to a pointwise error proportional to $N^{-1/2}$ (or $L^{-\frac{3}{2}}$). However, as previously remarked,   these estimates are not enough to justify the use of the exact closed shell formulas on the analysis of the next-to-leading order term in the asymptotic expansion for the exact exchange. To improve on this error estimate and include the more general rectangular box case, we use our own estimate in Lemma \ref{central}. More precisely, setting $\alpha =0$ and $z=0$ in Lemma \ref{central}, we obtain\footnote{Note that for the case of the cubic box, finding the optimal (algebraic) coefficient on the remainder goes under the name of sphere problem and has been studied by many authors \cite{Vinogradov1963,chen1963improvement,Chamizo1995,HeathBrown1999}. In particular, better estimates (with smaller exponent) than \eqref{generalsphereest} are available in this case.}
\begin{align} |\Z^3\cap B^D_R|  - \frac{4\pi}{3}|D| R^3 = \mathcal{O}(R^{\frac{34}{23}+\epsilon}). \label{generalsphereest} \end{align}
As a consequence, by adapting the arguments in \cite{Friesecke1997} we can prove the following lemma.
\begin{lemma}[Open shell control]\label{openshellcontrol}
Let $\alpha,\beta \in \N_0^3$ and $\epsilon>0$. Then, there exists a constant $c = c(\alpha,\beta,\epsilon) >0$ independent of $N$ and $L$ such that for any determinantal ground state of the free $N$-electron gas in $Q^D_L$ (under either Dirichlet, Neumann or periodic boundary conditions) we have
\begin{align} |\partial^\alpha_r \partial^\beta_{\tilde{r}} \gamma_{N,L}(r,\tilde{r})-\partial^\alpha_r \partial^\beta_{\tilde{r}} \gamma_{N_-,L}(r,\tilde{r})| \leq c \frac{N^{\frac{|\alpha|+|\beta|}{3}+\frac{34}{69}+\epsilon}}{L^{3+|\alpha|+|\beta|}}. \label{openshellcontroleq} \end{align}
In particular, if $\bar{\rho} = \frac{N}{|Q^D|L^3}$ is constant, one has 
\begin{align*} \partial^\alpha_r \partial^\beta_{\tilde{r}} \gamma_{N,L}(r,\tilde{r}) = \partial^\alpha_r\partial^\beta_{\tilde{r}} \gamma_{N_-,L}(r,\tilde{r}) + \mathcal{O}(L^{-\frac{35}{23}+\epsilon}) . \end{align*}
\end{lemma}
\begin{proof} For simplicity we disregard spin here. First, we note that by eq.~\eqref{generalsphereest}, the degeneracy of the open shell can be controlled by
\begin{align} d(N) = N_+-N_- \lesssim_{\epsilon} R_N^{\frac{34}{23}+\epsilon}\lesssim_{\epsilon} N^{\frac{34}{69}+\epsilon} ,\label{degeneracycontrol} \end{align}
where the $\epsilon$ in the second inequality is different from the first by a factor of $1/3$. Next, if $\Psi$ is a determinantal ground state of the free $N$-electron gas in $Q^D_L$, then we know (see section~\ref{sec:electrongasformulas}) that, up to a phase factor,
\begin{align*} \Psi = \underbrace{\biggr(\bigwedge_{k \in B_{R_{N_-}}} \phi^L_{k} \biggr)}_{= \Psi_{N_-,L}}\wedge \widetilde{\phi}_1 \wedge ... \wedge \widetilde{\phi}_{N-N_-}  ,\end{align*}
where $\wedge$ is the anti-symmetric tensor product (see \eqref{slaterdeterminant}) and $\{\tilde{\phi}_i \}_{i\leq N-N_-}$ is a set of orthonormal functions given by linear combinations of the orbitals in the open shell, i.e.
\begin{align*} \widetilde{\phi}_i = \sum_{\substack{k \in \N^3 \\ |k| = R_N}} c_{ik} \phi_{k}^L \quad \mbox{ for some $\{c_{ik}\}_{i,j} \subset \C$ satisfying $\sum_{k} c_{ik} \overline{c_{jk}} = \delta_{ij}$.} \end{align*}
In particular, $C^\ast = \{\overline{c_{ki}}\}$ is unitary (from $\C^{N-N_-}$ to $\C^{d(N)}$) and we have
\begin{align}
    \biggr|\sum_{i,j,k} c_{ij} \overline{c_{ik}} a_j b_k \biggr|\leq \biggr(\sum_j |a_j|^2\biggr)^\frac12 \biggr(\sum_k |b_k|\biggr)^\frac12  \label{eq:isometry}
\end{align}
for any $(a_1,..a_{d(N)}),(b_1,...,b_{d(N)}) \in \C^{d(N)}$. We can now use the formula
\begin{align*} \gamma_{\Psi}(r,\tilde{r}) = \underbrace{\sum_{k \in B_{R_{N_-}}\cap \N^3} \phi_k^L(r)\overline{\phi_k^L}(\tilde{r})}_{= \gamma_{N_-,L}(r,\tilde{r})} + \sum_{i=1}^{N-N_-} \widetilde{\phi}_i(r)\overline{\widetilde{\phi}_i}(\tilde{r}) \end{align*}
and the estimates \eqref{degeneracycontrol},\eqref{eq:isometry}, and $|\partial^\alpha \phi_{k}^L(r)| \leq c_\alpha (R/L)^{|\alpha|} L^{-\frac32} \leq c_\alpha N^{\frac{|\alpha|}{3}}L^{-\frac32 - |\alpha|}$ (see \eqref{eigenvectorsper}--\eqref{eigenvectorsneu}) to conclude that
\begin{align*}  |\partial^\alpha_r\partial^\beta_{\tilde{r}}\gamma_{\Psi}(r,\tilde{r})-\partial^\alpha_r\partial^\beta_{\tilde{r}}\gamma_{N_-,L}(r,\tilde{r})| &= \biggr|\sum_{i=1}^{N-N_-}\partial^\alpha \widetilde{\phi_i}(r) \partial^\beta \overline{\widetilde{\phi}_i}(\tilde{r})\biggr| = \biggr|\sum_{i,j,k} c_{ij} \overline{c}_{ik} \partial^\alpha \phi_j^L(r) \overline{\partial^\beta \phi_k^L}(\tilde{r}) \biggr| \\
&\leq \biggr(\sum_{|j|= R_N} |\partial^\alpha \phi_j^L(r)|^2\bigg)^\frac12 \biggr(\sum_{|k|=R_N} |\overline{\partial^\beta \phi_k^L}(\tilde{r})|^2\biggr)^\frac12 \\
&\lesssim \frac{N^{\frac{|\alpha|+|\beta|}{3}+\frac{34}{69}+\epsilon}}{L^{3+|\alpha|+|\beta|}}. \end{align*} \end{proof}
\begin{remark*} One could equally well consider the density matrix of the current closed shell $\gamma_{N_+,L}$ in Lemma \ref{openshellcontrol}.\end{remark*}
\subsection{Fermi momentum asymptotics}

Here we derive the asymptotics of the finite-size Fermi momentum appearing in \eqref{layerapprox2}.

For the non-interacting electron gas model, the finite-size Fermi momentum is defined as the momentum of the highest occupied orbital of the ground state wave function. For the free $N$-electron gas in $Q^D_L$, it is simply given (in atomic units)  by
\begin{align} 
    p_{N,L} \coloneqq \begin{dcases} \frac{2\pi R^{\textnormal{Per}}_N}{L}, \quad &\mbox{ for periodic BCs,} \\ \frac{\pi R^{\textnormal{Dir}}_N}{L}, \quad &\mbox{ for Dirichlet BCs,} \\
    \frac{\pi R^{\textnormal{Neu}}_N}{L}, \quad &\mbox{ for Neumann BCs,} \end{dcases} \label{fermimomentumdef} 
\end{align}
where $R_N$ is the Fermi radius. It is well known that in the thermodynamic limit, the finite-size Fermi momentum converges to the (continuum) Fermi momentum, defined as
\begin{align*} p_F = (3\pi^2 \bar{\rho})^{\frac{1}{3}} . \end{align*}
The next lemma presents the next-order correction of the finite-size Fermi momentum, which is crucial for deriving the next-order corrections from Theorem \ref{mainthm}.
\begin{lemma}[Fermi momentum asymptotics] \label{latticeeffect} Let $\bar{\rho} = \frac{N}{|Q^D|L^3}$ be constant, % $R_{N_-}$ be the Fermi radius of the previous closed shell, 
$p_F$ be the Fermi momentum and $p_{N,L}$ be the finite-size Fermi momentum. Then
\begin{align}  p_{N,L} = \begin{dcases} p_F + \mathcal{O}(L^{-\frac{35}{23}+\epsilon}) , &\mbox{ for periodic BCs,} \\
  p_F+\frac{\pi|\partial Q^D|}{8|Q^D|}L^{-1} +\mathcal{O}(L^{-\frac{35}{23}+\epsilon}), &\mbox{ for Dirichlet BCs,} \\
  p_F - \frac{\pi|\partial Q^D|}{8|Q^D|} L^{-1} +\mathcal{O}(L^{-\frac{35}{23}+\epsilon}) , &\mbox{for Neumann BCs.} \end{dcases} \label{p_Lestdir}\end{align}
\end{lemma}
\begin{proof} We present the proof only for the Dirichlet case and denote $R^{\textnormal{Dir}}_{N}$ simply by $R_{N}$. First, from the previous section we already know that $R_N\lesssim N^{\frac{1}{3}}$ and $0\leq N_+-N\leq d(N) \lesssim N^{\frac{34}{69}+\epsilon}$. Since $|Q^D| = |D|$, one has
\begin{align*} \frac{\bar{\rho}L^3}{2} &= \frac{N}{2|D|} =  \frac{|\N^3\cap B^D_{R_{N}}|}{|D|} +\mathcal{O}(N^{\frac{34}{69}+\epsilon}) \\
&= \frac{1}{8|D|}(|\Z^3\cap B^D_{R_{N}}|-\sum_{j=1}^3 |\Z^2 \cap \pi_j(B_{R_{N}}^D)|) +\mathcal{O}(N^{\frac{34}{69}+\epsilon}) \\
&= \frac{\pi}{6} R_{N}^3 - \frac{\pi |\partial Q^D|}{16|D|} R_{N}^2 + \mathcal{O}(N^{\frac{34}{69}+\epsilon}) , \end{align*}
which implies that
\begin{align} 
p_F^3 = p_{N,L}^3 - \frac{3\pi |\partial Q^D|}{8|D|}\frac{1}{L} p_{N,L}^2 + \frac{\mathcal{O}(N^{\frac{34}{69}+\epsilon})}{L^3} . \label{auxeq} \end{align}
Then, with $\bar{\rho} = \frac{N}{|Q^D|L^3}$ fixed, it is clear that $p_{N,L}^3 \ra p_F^3$. The correction proportional to $1/L$ now follows from  equation \eqref{auxeq} by Taylor expansion of the (near the solution $p_F$ for $1/L=0$ smooth) solution of the cubic equation for $p_{N,L}$.  \end{proof}

\begin{remark*} The above lemma can be viewed as a variant of the famous two term Weyl's law \cite{ivrii2016100} on the asymptotic behaviour of eigenvalues of the Laplacian (with an improved remainder). \end{remark*}

\section{Discrete to continuum approximation} \label{sec:continuumapprox}
In this section, we present the proof of Lemma \ref{central} and use it to derive the continuum approximation for the density  matrix with explicit estimates. To simplify the notation and make the proofs more efficient, we will focus on the case of the unit cube $Q$ and point out the modifications necessary for the more general domain $Q^D$ at the end of each proof.

\subsection{Estimate on exponential sums}

Here we want to derive non-trivial estimates\footnote{the trivial estimate being $|\sum_{k \in \Z^n \cap S} e^{i f(k)}| \leq |\Z^n \cap S|$, which holds for any subset $S \subset \R^n$ and any real-valued function $f(k)$}  for some weighted exponential sums that appear naturally in the proof of Lemma \ref{central}. More precisely, let $\alpha \in \N_0^3$, $M,R>0$ and $Q^h_M$ be the cubic holed box defined by
\begin{align} Q^h_{M} \coloneqq \{ r\in \R^3 : M < |r|_{max}\leq \frac{4}{3} M \},  \label{holedboxdef} \end{align}
then our goal is to find a better than trivial estimate for the sum
\begin{align}  S^\alpha_{M}(R,z) \coloneqq \sum_{k \in Q^{h}_{M}\cap (2\pi\Z)^3} \frac{(k+z)^{\alpha}}{|k+z|^{|\alpha|+2}}e^{i R |k+z|} , \label{mainsumdef} \end{align}
when $M$ is large and $R \sim M^2$. Using the above notation the main estimate can be stated as follows.
\begin{lemma} Let $\epsilon>0$, $\alpha \in \N_0^3$ and $z \in \R^3$ with $|z|_{\max} \leq \pi$. Then
\[ S_{M}^\alpha(R,z) \lesssim_\epsilon M^{\epsilon}(R^{\frac{1}{12}} M^{\frac{3}{4}} + M^{\frac{11}{12}} + R^{-\frac{1}{24}} M^{\frac{23}{24}}) , \]
where the implicit constant depends on $\epsilon$ and $\alpha$, but is independent of $M, R$ and $z$.
\label{exponentiallemma}
\end{lemma}

Before proving this result, we recall a recent improvement due to Heath-Brown \cite{HeathBrown2017} of the classical Van der Corput $k^{th}$-derivative estimate \cite{graham_kolesnik_1991} on exponential sums.
\begin{theorem}[Heath-Brown \cite{HeathBrown2017}] \label{heathbrownthm} Let $k\geq 3$ be an integer and $f \in C^k([0,M],\R)$. Suppose that
\[ 0 < \lambda_k \leq f^{(k)}(s) \leq A\lambda_k , \quad s \in (0,M) ,\]
for some $A>0$. Then
\[ \sum_{n \leq M}e^{i f(n)} \lesssim_{k,A,\epsilon} M^{1+\epsilon}\Bigl(\lambda_k^{\frac{1}{k(k-1)}}+M^{-\frac{1}{k(k-1)}}+M^{-\frac{2}{k(k-1)}}\lambda_k^{-\frac{2}{k^2(k-1)}}\Bigr) .\]
\end{theorem}
\begin{remark*} Note that for any $c<1$, the estimate above still holds for the sum in the interval $cM \leq n \leq M$ (with two times the implicit constant above). Indeed, one can simply consider the sum in this interval as the sum until $M$ minus the sum until $cM$, where both can be controlled by the same factor.\end{remark*} 

\noindent We also recall an elementary partial summation lemma (see e.g. \cite{graham_kolesnik_1991,Huxley1996} for the proof) that is used to deal with the weight function in $S^\alpha_M(R,z)$.
\begin{lemma}[Partial summation \cite{Huxley1996}]\label{partialsum} Let $g \in C^1([a,b])$ and denote the total variation $\int_a^b |g'|$ of $g$ by $V_{g}[a,b]$. Then for any sequence $a_n$, one has
\[\biggr|\sum_{n\geq a}^b g(n)a_n \biggl| \leq (V_g[a,b]+|g(a)|)\max_{\gamma \leq b} \biggr|\sum_{\gamma}^b a_n\biggl| . \]
\end{lemma}

\begin{proof}[Proof of Lemma \ref{exponentiallemma}] First note that if $M$ is small (compared to $R$), then the estimate is trivial. Therefore, we need to consider only the case where $M$ is big. Next, defining
\begin{align}
f(k) \coloneqq R|k+z| \quad \mbox{ and } \quad  g(k) \coloneqq \frac{(k+z)^\alpha}{|k+z|^{2+|\alpha|}} , \label{fgdef} \end{align}
the idea is to see $f$ and $g$ as functions of one coordinate, estimate the inner sum by Theorem \ref{heathbrownthm} and Lemma \ref{partialsum}, and then use the trivial estimate for the outer sums. To verify the assumptions of Theorem \ref{heathbrownthm}, let us consider the sets
\begin{align} U^j_M \coloneqq \{k\in \R^3:|k_j|\geq M \mbox{ and } \max_{\ell \neq j}\{|k_\ell|\} \geq M/10 \}\cap Q_M^h ,\label{Usetdef} \\
L^j_M \coloneqq \{k\in\R^3: |k_j|\geq M \mbox{ and } \max_{\ell \neq j}\{|k_\ell|\} < M/10 \}\cap Q_M^h  , \label{Lsetdef} \end{align}
and observe that a straightforward calculation yields
\begin{align} \partial_{k_j}^4 f(k) = -3 R \frac{|\pi_j (k+z)|^2(|\pi_j (k+z)|^2-4(k_j+z_j)^2)}{|k+z|^7} . \label{derivative} \end{align}
Now the reason for our choice of $\tfrac43$ in the definition of $Q^h_M$ (any number between $1$ and $\sqrt{2}$ would be enough) is that, since $|z|_{max}\leq \pi$ is small compared to $M$, inside of $U_M^j$ we have 
\begin{align*}4(k_j+z_j)^2-|\pi_j(k+z)|^2 \sim M^2 ,  \quad \quad |k+z| \sim M , \quad \mbox{ and }\quad |\pi_j (k+z)|^2 \sim M^2 \end{align*}
(with implicit constants independent of $M$ and $R$).
In particular, from \eqref{derivative}, we find that
\begin{align} |\partial^{4}_{k_j} f(k)| \sim \frac{R}{M^3} \label{fderest} \end{align}
inside $U^j_M$. Moreover, note that $|g(k)|\lesssim \frac{1}{|k|^2}$ and $|\partial_{k_j} g(k)|\lesssim \frac{1}{|k|^3}$. Hence, as $|k| \geq M$ for any $k \in Q_M^h$, if we consider $g$ as a function of only one coordinate, say $k_3$, then for any $\{k_1\}\times \{k_2\}\times I \subset Q_M^h$ we have
\begin{align} V_{g(k_1,k_2,\cdot)}[I] + |g(k)| \lesssim M^{-2}  \label{gvar} \end{align}
with implicit constant independent of $k_1,k_2,I, M$ and $R$. Therefore, we can apply Lemma \ref{partialsum} and Theorem \ref{heathbrownthm} (see  \eqref{fderest},\eqref{gvar}) to conclude that,  for any $I_1\times I_2\times I_3 \subset U_M^j$,
\begin{align} \sum_{\substack{k\in I_1\times I_2\times I_3 \\ k \in (2\pi\Z)^3}} g(k) e^{i f(k)} &\lesssim_\epsilon \sum_{\substack{ k_\ell \in 2\pi\Z\cap I_\ell \\ \ell\neq j}} M^{-2} M^{1+\epsilon}(R^{\frac{1}{12}} M^{-\frac{1}{4}} + M^{-\frac{1}{12}} + R^{-\frac{1}{24}} M^{-\frac{1}{24}}) \nonumber \\
&\lesssim M^\epsilon (R^{\frac{1}{12}} M^{\frac{3}{4}} + M^{\frac{11}{12}} + R^{-\frac{1}{24}} M^{\frac{23}{24}}) .
\label{intervalsumest} \end{align}
On the other hand, if $I_1\times I_2\times I_3 \subset L_M^j$, then from equation \eqref{derivative} and the definition of $L_M^j$ (see \eqref{Lsetdef}), one can show that $\partial_{k_\ell}^4 f(k) \sim \frac{R}{M^3}$ for $\ell\neq j$. In particular, by summing over $I_\ell$ first and using the same arguments as before, we conclude that, as long as $I_1\times I_2\times I_3 \subset L_M^j$, the sum over $I_1\times I_2\times I_3$ is again bounded by the right hand side of \eqref{intervalsumest}.

Finally, note that we can split the summation over $Q_M^h$ as
\begin{align*} \sum_{k \in Q_M^h\cap (2\pi\Z)^3} g(k) e^{i f(k)} = \sum_{p=1}^{P} \biggr(\sum_{\substack{ k \in I_p^1\times I_p^2 \times I_p^3 \\ k \in (2\pi\Z)^3}} g(k) e^{i f(k)} \biggl) , \end{align*}
where the $I_p^j$ are intervals such that the product $I_p^1\times I_p^2 \times I_p^3$ is contained in one of the $U^j_M$ or $L^j_M$, and the number $P$ is independent of $M$. Therefore, the result follows from estimating the sums for each $p$ independently by \eqref{intervalsumest} and summing the estimates up (since $P$ is independent of $M$).\end{proof}

\begin{remark*}[on the proof above]
\begin{itemize}
\item The important point in the estimate above is that it gives better than trivial estimates for $M\approx \sqrt{R}$, which is crucial for improving the exponent 3/2 in \eqref{centralest}.
\item The classical Van der Corput $4^{th}$-derivative estimate (see \cite{graham_kolesnik_1991}) would actually be enough for our purposes. However, Theorem \ref{heathbrownthm} gives a slightly better estimate on the error term. 
\item If we assume that $z \in \Q^3$, then one can show that the same estimate from \cite[Lemma 3.1]{Chamizo1995} holds, which would then lead to the Vinogradov, Chen \cite{Vinogradov1963,chen1963improvement} exponent of $4/3$ in \eqref{centralest}. Unfortunately, in this case the implicit constant depends on the least common multiple of the denominators of $z_1,z_2,z_3$, and therefore, we lose the uniform control that is essential for our purposes in the next section. \end{itemize}\end{remark*}

\begin{remark*}[Generalization to $Q^D$] For general $Q^D$ we need an analog of Lemma \ref{exponentiallemma} with $f$ and $g$ from equation \eqref{fgdef} replaced by $f(k) = R|D(k+z)|$ and $g(k) = \frac{(D(k+z))^\alpha}{|D(k+z)|^{|\alpha|+2}}$. In particular, the only significant modifications are that equation \eqref{derivative} has to change accordingly, and the sets $U^j_M$ and $L^j_M$ defined in \eqref{Usetdef} and \eqref{Lsetdef} have to be replaced by their pre-images under $D$ (as a map in $\R^3$). \end{remark*}

\subsection{Proof of Lemma \ref{central}}
To complete the proof of Lemma \ref{central}, we shall use two classical results. The first one is the celebrated Poisson summation formula (see e.g. \cite{sogge2017fourier}), already used in the present context in \cite{Friesecke1997}, which we write in a slightly different form to fit our goal.
\begin{lemma}[Poisson summation formula] Let $u\in C^\infty_c(\R^n)$, then (see the Fourier transform convention adopted \eqref{fourierconvention}) the Poisson summation formula states that
\begin{align} \sum_{k\in \Z^n} \frac{1}{(2\pi)^n}\widehat{u}(k) e^{i k\scpr z} - u(z) = \sum_{\substack{k\in (2\pi\Z)^n \\ k \neq 0}} u(z+k) . \label{poisson} \end{align}
\end{lemma}
The second one is an (optimal) estimate on the decay of the Fourier transform of the characteristic function of the ball and its derivatives. The proof is a straightforward calculation (at least in $\R^3$).
\begin{lemma}[Fourier transform of the ball] \label{fouriertransformball} Let $\chi_{B_1}$ be the characteristic function of the unit ball in $\R^3$. Then
\[ \widehat{\chi_{B_1}}(k) = \frac{4 \pi}{3} h(|k|) , \]
where $h(s) =  3(\sin s - s \cos s)/s^3$. Moreover, for any $\alpha \in \N_0^3$, there exists $c_\alpha>0$ such that
\begin{align} |\partial^\alpha \widehat{\chi_{B_1}}(k)| \leq \frac{c_{\alpha}}{(1+|k|)^2}, \quad \forall k \in \R^3 .\label{fourierdecay} \end{align}
\end{lemma}
\begin{proof}[Proof of Lemma \ref{central}] Let $R>0$, then our goal now is to bound the error term
\begin{align*} \mathcal{E}^{\alpha}(R,z) \coloneqq \sum_{k\in B_R\cap \Z^3} (i k)^\alpha e^{i k\scpr z} - \int_{B_R} (i k )^\alpha e^{i k \scpr z} dk .\end{align*}
For this, the main idea is to apply the Poisson summation formula to a smooth version of the characteristic function of the ball of radius $R$, use Lemma \ref{exponentiallemma} to control the error on the Fourier side and then estimate the difference between our smoothed error and the error $\mathcal{E}^\alpha$ defined above. So first, let $H>0$ be a small parameter to be chosen later, pick any non-positive function $\phi \in C_c^\infty(0,1)$ with $\int \phi(s)ds = -1$, and define
\begin{align} f_{R,H}(r) = \begin{cases} 1, &\mbox{if}\quad |r| \leq R,\\ \displaystyle{1+ \int_0^{|r|-R} \phi^H(s) \mathrm{d}s},  &\mbox{otherwise,}\end{cases} \label{smoothed} \end{align}
where $\phi^H(s) \coloneqq \frac{1}{H} \phi\bigr(\frac{s}{H}\bigr)$. One can check that $f_{R,H}$ is smooth, assumes only values between 0 and 1, and vanishes outside $B_{R+H}$. Moreover, observing that $\widecheck{f}_{R,H}(k) = \frac{1}{(2\pi)^3}\widehat{f}_{R,H}(k)$ (as $f_{R,H}$ is symmetric) and applying the Poisson summation formula to $\widehat{u} = f_{R,H}$ we obtain
\begin{align} \underbrace{\sum_{k\in \Z^3} f_{R,H}(k)e^{i k\scpr z} - \widehat{f}_{R,H}(z)}_{\coloneqq \mathcal{E}_H(R,z)} = \underbrace{\sum_{\substack{k \in (2\pi\Z)^3 \\ k \neq 0}} \widehat{f}_{R,H}(z+k)}_{\coloneqq \widehat{\mathcal{E}}_H(R,z)} . \label{Poisson}\end{align}
The next step is to use Lemma \ref{exponentiallemma} to estimate the right-hand side of \eqref{Poisson}. For this, we first use spherical coordinates to obtain 
\begin{align*} \widehat{f}_{R,H}(z+k) &= \int_{\R^3} f_{R,H}(r) e^{-i r \scpr (z+k)} \mathrm{d}r\\
&= 4\pi \int_0^R \frac{\tau \sin(|k+z| \tau)}{|k+z|}
\mathrm{d}\tau +4\pi \int_R^{R+H} \frac{\tau \sin(|k+z| \tau)}{|k+z|} \biggr(1+ \int_0^{\tau-R} \phi^H(s) ds\biggl) \mathrm{d}\tau \\
 &= -4\pi \int_R^{R+H} \frac{\sin(|k+z|\tau)-(|k+z|\tau)\cos( |k+z|\tau)}{|k+z|^3} \phi^H(\tau-R) \mathrm{d}\tau .\end{align*}
In particular, by partial integration and re-scaling we have
\begin{align} \widehat{\mathcal{E}}_H(R,z) &= -4\pi \int_R^{R+H} \biggr(3\dot{\phi}^H(\tau-R) + \tau\ddot{\phi}^H(\tau-R)\biggr) \sum_{\substack{k \in (2\pi\Z)^3 \\ k \neq 0}} \frac{ \cos(|k+z|\tau)}{|k+z|^4} \mathrm{d}\tau \nonumber \\
&= -4\pi \int_0^1 \biggr(\overbrace{\frac{3}{H}\dot{\phi}(\tau)+\frac{\tau}{H}\ddot{\phi}(\tau)}^{= \mathcal{O}(\frac{1}{H})} + \frac{R}{H^2} \ddot{\phi}(\tau)\biggr) \sum_{\substack{k \in (2\pi\Z)^3 \\ k \neq 0}} \frac{\cos\bigr((R+H\tau)|k+z|\bigr)}{|k+z|^4} \mathrm{d}\tau \label{middle} \\
&= \mathcal{O}\biggr(\frac{1}{H}\biggr) - \frac{4\pi R}{H^{2}}\int_0^1\ddot{\phi}(\tau) \sum_{\substack{k \in (2\pi\Z)^3 \\ k \neq 0}} \frac{\cos\bigr((R+H\tau)|k+z|\bigr)}{|k+z|^4} \mathrm{d}\tau , \nonumber \end{align}
where for the last equality we used that $\sum_{k\in (2\pi\Z)^3\setminus \{0\}} |k+z|^{-4} \leq c$ (since $|z|_{\max} \leq \pi$). Here we slightly abused the big-O notation to denote terms bounded in absolute value by a constant times $H^{-1}$ by $\mathcal{O}(H^{-1})$. Now let $M>0$, then by integrating by parts $n$ times, throwing out the terms outside the box $Q_M\setminus Q_1 = [-M,M]^3\setminus [-1,1]^3$,  and integrating back by parts (recall that $\phi \in C^\infty_c(0,1)$), we find that
\begin{align*} 
    \widehat{\mathcal{E}}_H(R,z) &= 4\pi R \int_0^1 \phi(\tau) \sum_{\substack{k \in (2\pi\Z)^3 \\ \cap Q_M\setminus Q_1 }} \frac{\cos\bigr((R+H\tau)|(k+z)|\bigr)}{|k+z|^2} \mathrm{d}\tau +\mathcal{O}\biggr(\frac{1}{H} + \frac{R}{H^{n+1}M^{n}}\biggr). 
\end{align*}
Next, after possibly replacing the implicit constant above by a factor independent of $M$, $H$, and $R$, we can assume that $M = (\frac{4}{3})^q$ for some $q\in \N$. Thus by using the decomposition $Q_M\setminus Q_1 = \bigcup_{k=1}^{q-1} Q_{(\frac{4}{3})^k}^h$, the facts that $H<1$ and $|z|_{max}\leq \pi$ (by periodicity), and Lemma \ref{exponentiallemma}, we find that
\begin{align}
   \widehat{\mathcal{E}}_H(R,z) \lesssim_{\epsilon,n}&  \frac{1}{H} + \frac{R}{H^{n+1}M^n} + \sum_{k=1}^{q-1} R \biggr(\frac{4}{3}\biggr)^{k\epsilon} \biggr(R^{\frac{1}{12}}\biggr(\frac{4}{3}\biggr)^{\frac{3k}{4}}+\biggr(\frac{4}{3}\biggr)^{\frac{11k}{12}}+ R^{-\frac{1}{24}}\biggr(\frac{4}{3}\biggr)^{\frac{23k}{24}}\biggr) \nonumber \\
    \lesssim_{\epsilon,n}& \frac{1}{H} + \frac{R}{H^{n+1}M^n} + \log M \biggr( R^{\frac{13}{12}}M^{\frac{3}{4}+\epsilon}+RM^{\frac{11}{12}+\epsilon}+R^{\frac{23}{24}}M^{\frac{23}{24}+\epsilon}\biggr) ,\label{est1} \end{align}
where the $\log M$ can be absorbed in $M^\epsilon$. Moreover, we can estimate the error coming from the smoothing procedure as
\begin{align}
    |\mathcal{E}^0(R,z) - \mathcal{E}_H(R,z)| \leq \!\!\sum_{\substack{k \in \Z^3\cap B_{R+H} \\ k \not \in B_R}} \!\!\!\!\!\! |f_{R,H}(k)| + |\widehat{f}_{R,H}(z)-\widehat{\chi_{B_R}}(z)|  \lesssim R^2 H. \label{est2} 
\end{align}
Hence by summing estimates \eqref{est1} and \eqref{est2}, taking into account \eqref{Poisson}, setting $M = R^m$ and $ H = R^{-h}$, and minimizing the exponents, one concludes that
\begin{align*} 
    \mathcal{E}^0(R,z) \lesssim_{\epsilon,n} R^{2-h(n)+\epsilon m(n)} \log R  , \quad \forall n\in \N , 
\end{align*}
where $h(n) \coloneqq \frac{12n+11}{23n+22}$ and $m(n)\coloneqq \frac{12n + 12}{23 n + 22}$. Thus, since $h \uparrow \frac{12}{23}$ for $n\ra \infty$, by choosing  $n$ big enough, the result for $\alpha = 0$ follows.

For the estimates with $\alpha\neq 0$, we repeat the same arguments by taking $f^\alpha_{R,H}(r) = (ir)^\alpha f_{R,H}(r)$. In this case, note that by the Leibniz rule and using that $\partial^\alpha (|r|^{-4})\lesssim |r|^{-4-|\alpha|}$ away from the origin, we have
\begin{align} 
    \partial_k^\alpha \biggr( \frac{\cos((R+H\tau)|k+z|)}{|k+z|^{4}}\biggl) - \frac{\bigr(iR(k+z)\bigr)^\alpha}{2|k+z|^{|\alpha|+4}} \biggr(e^{i(R+H\tau)|k+z|}+(-1)^{|\alpha|} e^{-i(R+H\tau)|k+z|}\biggr) \nonumber \\
    \lesssim \frac{R^{|\alpha|-1}}{|k+z|^4}. \label{cosderest} \end{align}
We also have, by Poisson summation and the symmetry of $f^\alpha_{R,H}$, that
\begin{align} \label{Poisson'}\underbrace{\sum_{k\in \Z^3} f^\alpha_{R,H}(k)e^{i k\scpr z} - \widehat{f}^\alpha_{R,H}(z)}_{\coloneqq \mathcal{E}^\alpha_H(R,z)} = \underbrace{\sum_{\substack{k \in (2\pi\Z)^3 \\ k \neq 0}} \widehat{f}^\alpha_{R,H}( z+k)}_{\coloneqq \widehat{\mathcal{E}}^\alpha_H(R,z)} . \end{align}
Hence using the identity $\widecheck{(i \cdot)^\alpha f} = \partial^{\alpha} \widecheck{f}$ together with estimate \eqref{cosderest}, and repeating the same steps from before (see \eqref{middle}), we conclude that
\begin{align*} 
    \mathcal{E}^\alpha_H(R,z) &= -4\pi \int_0^1 \biggr(\frac{3}{H}\dot{\phi}(\tau)+\frac{R+H\tau}{H^2}\ddot{\phi}(\tau)\biggr) \sum_{\substack{k \in (2\pi\Z)^3 \\ k \neq 0}}\partial^\alpha \biggr(\frac{\cos((R+H\tau)|k+z|)}{|k+z|^4}\biggl) \\
&= -\frac{4\pi R}{H^2} \int_0^1 \ddot{\phi}(\tau) \!\!\sum_{\substack{k \in (2\pi\Z)^3  \\ k\neq 0}} \!\! \frac{\bigr(i R(k+z)\bigr)^\alpha}{2|k+z|^{|\alpha|+4}} \biggr(e^{i(R+H\tau)|k+z|}+(-1)^{|\alpha|} e^{-i(R+H\tau)|k+z|}\biggr) \mathrm{d}\tau  \\
&+\mathcal{O}\biggr(\frac{R^{|\alpha|}}{H^2} \biggr) . \end{align*}
Thus we can again integrate by parts $n$ times, remove all terms for $|k|\not \in Q_M\setminus Q_1$ and use Lemma \ref{exponentiallemma} to show that
\begin{align} 
    \mathcal{E}^\alpha_H(R,z) \lesssim_\epsilon R^{|\alpha|}\biggr(\frac{1}{H^2} + \frac{R}{H^{n+1}M^n} + R^{\frac{13}{12}}M^{\frac{3}{4}+\epsilon}+RM^{\frac{11}{12}+\epsilon}+R^{\frac{23}{24}}M^{\frac{23}{24}+\epsilon}\biggr) . \label{middle2} 
\end{align}
Finally, by combining \eqref{middle2} with \eqref{Poisson'} and the following simple estimate for the smoothing error
\[
|\mathcal{E}^\alpha(R,z)-\mathcal{E}^\alpha_H(R,z)|\lesssim R^{2+|\alpha|}H \]
one obtains an overall bound which is just $R^{|\alpha|}$ times the right hand sides of \eqref{est1} plus \eqref{est2}, and 
the result follows by optimizing the exponent and using the same $n\ra \infty$ argument as before. \end{proof}

\begin{remark*}[Generalization to $Q^D$] To deal with the general case, one has to replace $f_{R,H}$ by $f^{D}_{R,H}(r) \coloneqq f_{R,H}(D^{-1} r)$ and observe that $\widehat{f}^D_{R,H}(k) = |D| \widehat{f}_{R,H}(D k)$. The rest of the proof follows from the same steps by using the generalization of Lemma \ref{exponentiallemma} discussed in the previous remark.
\end{remark*}

\subsection{Continuum approximation of density matrices}

We can now introduce continuum approximations for the density matrices and prove that they accurately approximate the exact ground state density matrices.

Recalling formula \eqref{densitymatrix} and replacing the sum over $k$ by an integral, we can define the continuum version of the density matrix as
\begin{align}
    \gamma_{N,L}^{\textnormal{ctm}}(r,\tilde{r}) = \begin{dcases}
    \bar{\rho}\frac{p_{N,L}^3}{p_F^3} h(p_{N,L}|r-\tilde{r}|_{D,L}), &\mbox{ for periodic BCs,} \\
    \bar{\rho}\frac{p_{N,L}^3}{p_F^3} \sum_{\sigma\in G} \det \sigma \, h(p_{N,L}|r-\sigma\tilde{r}|_{D,2L}) , &\mbox{ for Dirichlet BCs,} \\  \bar{\rho}\frac{p_{N,L}^3}{p_F^3} \sum_{\sigma\in G} h(p_{N,L}|r-\sigma\tilde{r}|_{D,2L}), &\mbox{ for Neumann BCs ,} 
    \end{dcases} \label{ctmdensitymatrix}
\end{align}
where $p_{N,L}$ and $p_F$ are the finite-size and continuum Fermi momentum (see \eqref{fermimomentumdef}), the function $h$ is defined in \eqref{hdef}, and $|r|_{D,2L} = |r \mod D(2L\Z)^3|$ (the torus distance). 
In fact, in the Dirichlet case this approximation was already introduced in \cite{Friesecke1997} but the error estimates obtained there were insufficient to conclude that the approximation is accurate enough to reveal the surface correction to Dirac exchange. 

The following new estimate will be sufficient to obtain the surface correction. It is a direct consequence of Lemma \ref{central}.
\begin{theorem}[Continuum approximation estimates]
\label{continuumapproxthm} Let $\bar{\rho}=\frac{N}{|Q^D|L^3}$ be fixed. Let $\Psi_{N,L}$ be any determinantal ground state of the free $N$-electron gas on $Q^D_L$, and let $\gamma_{N,L}$ be its single-particle density matrix introduced in \eqref{densitymatrixdef}. Then for any $\epsilon>0$ we have
\begin{align}
    | \gamma_{N,L}(r,\tilde{r})- \gamma_{N,L}^\textnormal{ctm}(r,\tilde{r})| \lesssim_{\epsilon}  \bar{\rho} N^{-\frac{35}{69}+\epsilon} \lesssim_\epsilon \bar{\rho}^{\frac{34}{69}+\epsilon} L^{-\frac{35}{23}+3\epsilon}. \label{densitymatrixest0}
\end{align}
More generally, for any $\alpha, \beta \in \N_0^3$,
\begin{align} |\partial_r^\alpha \partial_{\tilde{r}}^\beta \gamma_{N,L}(r,\tilde{r})-\partial^\alpha_r \partial_{\tilde{r}}^\beta \gamma_{N,L}^\textnormal{ctm}(r,\tilde{r})| \lesssim_{\epsilon,\alpha,\beta} \bar{\rho}^{1+\frac{|\alpha|+|\beta|}{3}} N^{-\frac{35}{69}+\epsilon} \lesssim_{\epsilon,\alpha,\beta} \bar{\rho}^{\frac{34}{69}+\frac{|\alpha|+|\beta|}{3}+\epsilon} L^{-\frac{35}{23}+3\epsilon}. \label{densitymatrixest} \end{align}
Moreover, one has
\begin{align} &|\gamma^{\textnormal{Dir/Neu}}_{N,L}(r,\tilde{r})| \lesssim_\epsilon\bar{\rho} \biggr(N^{-\frac{35}{69}+\epsilon} + (1+|r-\tilde{r}|_{D,2L})^{-2}\biggl) , \label{densitymatrixdecay} \\
 &|\gamma^{\textnormal{Per}}_{N,L}(r,\tilde{r})|\lesssim_\epsilon \bar{\rho} \biggr(N^{-\frac{35}{69}+\epsilon}+(1+|r-\tilde{r}|_{D,L})^{-2}\biggl) .\label{densitymatrixdecayper} \end{align}
\end{theorem}

\begin{proof}
According to Lemma \ref{openshellcontrol}, up to an error $\lesssim \bar{\rho}^{\frac{34}{69}+\epsilon} L^{-\frac{35}{69}+3\epsilon}$, we can use the closed shell formulas for any determinantal ground state.

Next, we know that (i) $R \sim L$, (ii) $\widehat{\chi}_{B^D_R}(z) = |D|R^3 \widehat{\chi}_{B_1}(R D z)$, (iii) $\widehat{\chi}_{B_1}(z) = \frac{4 \pi}{3} h(|z|)$, and (iv) the sums $\sum_{k\in \Z^3\cap B_R} e^{i\frac{\pi}{L} k\scpr D^{-1}w}$ and $\sum_{k\in \Z^3\cap B_R} e^{i \frac{2\pi}{L} k \scpr D^{-1} w}$ are periodic (in $w$) with respect to $D((2L\Z)^3)$ and $D((L\Z)^3)$ respectively. Hence, estimate \eqref{densitymatrixest0} follows directly from the closed shell formulas \eqref{densitymatrix} by applying Lemma \ref{central} with $z = \frac{\pi}{L} D^{-1}(r-\sigma \tilde{r}) \mod (2\pi\Z)^3$ for the Dirichlet and Neumann case and $z = \frac{2\pi}{L} D^{-1} (r-\tilde{r}) \mod (2\pi \Z)^3$ for the periodic case.

For the derivative estimates in \eqref{densitymatrixest}, one can simply use Lemma \ref{central} together with the identity $\widehat{(-i\cdot )^\alpha f} = \partial^\alpha \widehat{f}$. Note that each derivative gives an additional factor of $1/L$ which compensates for the $R^{|\alpha|}$ gained in \eqref{centralest} and accounts for the factor  $\bar{\rho}^{\frac{|\alpha|}{3}}$ in \eqref{densitymatrixest}. The decay estimates \eqref{densitymatrixdecay} and \eqref{densitymatrixdecayper} follow from estimate \eqref{densitymatrixest0} and the decay of $h$ (see Lemma \ref{fouriertransformball}).
\end{proof}

\section{Proof of Theorem \ref{mainthm}} \label{sec:exchangethm}

In this section we make the following simplification to the current notation: as we are dealing with the thermodynamic limit ($N,L\ra \infty$ with $\bar{\rho} = \frac{N}{|Q^D|L^3}$ fixed), all functions and constants depending on both $N$ and $L$ will simply be indexed by $L$ (e.g., $\rho_L = \rho_{N,L}$). 

\subsection{Semi-local functionals}
Our goal now is to prove the following two-term asymptotic expansion for general semi-local functionals of the density and its gradient.
\begin{theorem}[Asymptotics for general semi-local functionals]\label{generalasymp} Let $\frac{N}{|Q^D|L^3} = \bar{\rho} = constant$, $\Psi_{L}$ be any determinantal ground state of the free $N$-electron gas in $Q^D_L$, and $\rho_{L}$ be the associated single-particle density. Suppose that $f(a,b) \in C^0([0,\infty)^2)\cap C^1((0,\infty)\times[0,\infty))$ and let $F[\rho_L] \coloneqq \int_{Q^D_L} f(\rho_L(r),|\nabla \rho_L(r)|) \mathrm{d}r$. Then for any $\epsilon>0$ we have
\begin{align}
    F[\rho_L] = \begin{dcases} f(\bar{\rho},0)|Q^D|L^3 +  \mathcal{O}(L^{-\frac{34}{23}+\epsilon}) , &\mbox{ for periodic BCs,} \\
     f(\bar{\rho},0)|Q^D|L^3 + \biggr(c^{\textnormal{Dir}}_{BL}(\bar{\rho})+c_{FM}(\bar{\rho})\biggr)|\partial Q^D|L^2 + \smallO(L^2), &\mbox{ for Dirichlet BCs,} \\
     f(\bar{\rho},0)|Q^D|L^3 + \biggr(c^{\textnormal{Neu}}_{BL}(\bar{\rho})-c_{FM}(\bar{\rho})\biggr)|\partial Q^D|L^2 + \smallO(L^2), &\mbox{ for Neumann BCs,} \end{dcases} \nonumber
\end{align}
%\begin{align} &F[\rho^{\textnormal{Per}}_L] = f(\bar{\rho},0)|Q^D|L^3 +  \smallO(L^2) \label{asympper} \\
%&F[\rho^{\textnormal{Dir}}_L] = f(\bar{\rho},0)|Q^D|L^3 + \biggr(c^{\textnormal{Dir}}_{BL}(\bar{\rho})+c_{FM}(\bar{\rho})\biggr)|\partial Q^D|L^2 + \smallO(L^2) \label{asympdir} \\
% &F[\rho^{\textnormal{Neu}}_L] = f(\bar{\rho},0)|Q^D|L^3 + \biggr(c^{\textnormal{Neu}}_{BL}(\bar{\rho})-c_{FM}(\bar{\rho})\biggr)|\partial Q^D|L^2 + \smallO(L^2)  \label{asympneu} \end{align}
where the boundary layer and Fermi momentum corrections are given by
\begin{align*}
    &c_{BL}(\bar{\rho}) = \begin{dcases} \frac{1}{2p_F} \int_0^\infty f\bigr(\bar{\rho}(1-h(s)),2\bar{\rho}p_F|\dot{h}(s)|\bigr)-f(\bar{\rho},0) \mathrm{d} s, &\mbox{ for Dirichlet BCs,} \\
    \frac{1}{2p_F} \int_0^\infty f\bigr(\bar{\rho}(1+h(s)),2\bar{\rho}p_F|\dot{h}(s)|\bigr)-f(\bar{\rho},0) \mathrm{d} s , &\mbox{ for Neumann BCs,} \end{dcases}  \\
    &c_{FM}(\bar{\rho}) \coloneqq \frac{3\pi \bar{\rho}}{8p_F} \partial_1 f(\bar{\rho},0) .
\end{align*}
%\begin{align} &c^{\textnormal{Dir}}_{BL}(\bar{\rho}) \coloneqq \frac{1}{2p_F} \int_0^\infty f(\bar{\rho}(1-h(s)),2\bar{\rho}p_F|\dot{h}(s)|)-f(\bar{\rho},0) \mathrm{d} s  \\  &c_{BL}^{\textnormal{Neu}}(\bar{\rho}) \coloneqq \frac{1}{2p_F} \int_0^\infty f(\bar{\rho}(1+h(s)),2\bar{\rho}p_F|\dot{h}(s)|)-f(\bar{\rho},0) \mathrm{d} s \label{boundarylayerconst} \\
 %   &c_{FM}(\bar{\rho}) \coloneqq \frac{3\pi \bar{\rho}}{8p_F} \partial_a f(\bar{\rho},0) . \label{fermimomentumconst}\end{align}
\end{theorem}
\begin{proof} In the periodic case, since $f$ is differentiable near $(\bar{\rho},0)$ for any $\bar{\rho}>0$, the result follows from Theorem \ref{continuumapproxthm} and a simple Lipschitz estimate.

As the Neumann and Dirichlet case are  similar to each other, we give the details only for the Dirichlet case. The proof is somewhat technical, due to the fact that we have not assumed continuity of the derivative of $f$ with respect to $|\nabla \rho|$ at $\rho=0$ (otherwise the small parameter $\delta$ in the proof below could be set to zero). But allowing a derivative discontinuity at $\rho=0$ is essential to include the B88 functional, and shows the robustness of our results. 

First note that, since $f(\rho,|\nabla\rho|)$ depends only on $\rho$ and on the norm of $|\nabla \rho|$, by the symmetries pointed out in Lemma \ref{rhosymmetries} we can reduce the integration domain to
\begin{align*} F[\rho_L] = 8 \int_{Q^D_{\frac{L}{2}}} f(\rho_L,|\nabla \rho_L|) \mathrm{d}r . \end{align*}
Thus the continuum approximation from Theorem \ref{continuumapproxthm} reads
\begin{align} \rho^\textnormal{ctm}_L(r) = \gamma^\textnormal{ctm}_L(r,r) = \bar{\rho}\frac{p_L^3}{p_F^3}\sum_{\sigma \in G} \det \sigma h(p_{L} \norm{r-\sigma r}) \label{eq:ctmdensity0} \end{align}
(since $\norm{r-\sigma r}_{D,2L} = \norm{r-\sigma r}$ for any $r\in Q^D_{\frac{L}{2}}$ and $\sigma \in G$). The idea now is to use the continuum approximation to show that $\rho_L$ is only small close to the faces and  edges of the box. For this, first note that $h(r) =1$ if and only if $r=0$. Therefore, for any $\delta >0$ there exist $c(\delta)>0$ such that $1-h(2p_L |r_i|)\geq c(\delta)$ as long as $|r_i|>\delta$ and $L$ is big enough. Moreover, from \eqref{eq:ctmdensity0} and the fact that $h$ decays at infinity, we see that we can choose $R(\delta)>0$ such that
\begin{align}
    C\geq \rho^\textnormal{ctm}_L(r) \geq \frac{c(\delta)}{2} \quad \mbox{for any $r \in Q_{L,\delta}$ and some $C>0$}, \label{eq:ctmest}
\end{align}
where $Q_{L,\delta}$ is the region defined (using the projection $\pi_j$ from the notation section which removes the coordinate $r_j$) by
\begin{align*}
    Q_{L,\delta} \coloneqq \{ r\in Q^D_{\frac{L}{2}} : \min_{j\leq 3}|r_j|>\delta \quad \mbox{ and } \quad \min_{j\leq 3}|\pi_j r| > R(\delta) \}.
\end{align*}
Next, note that by the continuum approximation estimate (see equation \eqref{densitymatrixest}) and estimate \eqref{eq:ctmest}, we find that
\begin{align*}
    2C \geq \rho_L(r) \geq \frac{c(\delta)}{4} \quad \mbox{for any $r \in Q_{L,\delta}$ and $L$ big enough}.
\end{align*}
In particular, by the assumptions on $f$ ($C^1$ away from $a=0$ and $C^0$ up to $a=0$), there exists a Lipschitz constant $C(\delta)>0$ and a uniform (with respect to $\delta$ and $L$) constant $C_0>0$ such that
\begin{align*}
    |F[\rho_L]-F[\rho^\textnormal{ctm}_L]| &\lesssim \int_{Q^D_{\frac{L}{2}}\setminus Q_{L,\delta}} C_0 + C(\delta) \int_{Q_{L,\delta}}|\rho_L-\rho^\textnormal{ctm}_L|+|\nabla \rho_L - \nabla \rho^\textnormal{ctm}_L|   \\ 
    &\lesssim \delta L^2 + R(\delta)^2 L +  C(\delta) L^{\frac{34}{23}+\epsilon} . 
\end{align*}
Therefore, by dividing the above by $L^2$, taking the limit as $L\ra \infty$, and then the limit $\delta \ra 0$, one has
\begin{align}
    F[\rho_L] = F[\rho^\textnormal{ctm}_L] + \smallO(L^2) . \label{eq:ctmsmallo}
\end{align}

The next step is to work with the continuum versions on $Q_{L,\delta}$ and use a Taylor expansion of $f$ together with the decay of $h$ to determine the asymptotic coefficients. To shorten the notation, let us define
\begin{align*}
&f^L_k(r,t) \coloneqq \partial_k f\bigr(\bar{\rho}(1-t)+t\rho^\textnormal{ctm}_L(r),t|\nabla \rho^\textnormal{ctm}_L(r)|\bigr), \quad \quad  k \in \{1,2\},
\end{align*}
with the usual notation $\partial_kf$ for the partial derivative of $f$ with respect to the $k$-th argument. 
Then, note that by estimate \eqref{eq:ctmest} and the assumptions on $f$, we have
\begin{align}
    |f_{k}^L(r,t)| \leq C(\delta), \label{eq:fkest}
\end{align}
for any $(r,t)\in Q_{L,\delta}\times[0,1]$, for any $k\in\{1,2\}$, and for some $C(\delta)>0$. In addition, by the fundamental theorem of calculus we find
\begin{align} F[\rho^\textnormal{ctm}_L]
&=\underbrace{8\int_{Q_{\frac{L}{2}}^D\setminus Q_{L,\delta}} f(\rho^\textnormal{ctm}_L,|\nabla \rho^\textnormal{ctm}_L|) \mathrm{d}r}_{\lesssim \delta L^2 +R(\delta)^2 L} + 8\int_{Q_{L,\delta}} \biggr( f(\bar{\rho},0) 
+ \int_0^1 f^L_1(r,t)(\rho^\textnormal{ctm}_L(r)-\bar{\rho}) \mathrm{d}t \nonumber \\
&+ \int_{0}^1 \underbrace{f_2^L(r,t)|\nabla\rho_L^{\rm ctm}(r)|}_{=\frac{f_2^L(r,t)}{|\nabla \rho^\textnormal{ctm}_L(r)|} \sum_{j=1}^3 \bigl(\partial_{r_j} \rho^\textnormal{ctm}_L(r)\bigr)^2} \mathrm{d}t\biggr)\mathrm{d}r \nonumber \\
%&+\frac{f^L_{p}(r,t)}{|\nabla \rho^\textnormal{ctm}_L(r)|} \sum_{j=1}^3 \biggr(\frac{\bar{\rho}p_L^4}{p_F^3}\sum_{\sigma_{jj}=-1 } \det \sigma \dot{h}(p_L\norm{r-\sigma r}) \frac{4r_j}{\norm{r-\sigma r}}\biggr)^2 \mathrm{d}r \mathrm{d}t  \\
 &= \mathcal{O}(\delta L^2 + R(\delta)^2 L) + f(\bar{\rho},0)|Q^D|L^3 + \underbrace{8\bar{\rho}\int_0^1 \int_{Q_{L,\delta}}\!f^L_1(r,t) \biggr(\frac{p_L^3}{p_F^3}-1\biggr) \mathrm{d}r\mathrm{d}t}_{= :I(\delta)} \nonumber \\
 &+\sum_{\sigma \neq id} \underbrace{\frac{8\bar{\rho} p_L^3}{p_F^3}\int_0^1\int_{Q_{L,\delta}} \!f^L_1(r,t) \det \sigma h(p_L\norm{r-\sigma r}) \mathrm{d} r\mathrm{d}t}_{=: J_\sigma(\delta)}  \nonumber \\
 &+\sum_{j=1}^3 \underbrace{\frac{8\bar{\rho}^2p_L^6}{p_F^6}\int_0^1 \int_{Q_{L,\delta}} \frac{f^L_2(r,t)}{|\nabla \rho^\textnormal{ctm}_{L}(r)|}\biggr(\sum_{\sigma_{jj}=-1} \det \sigma p_L\dot{h}(p_L \norm{r-\sigma r}) \frac{4r_j}{\norm{r-\sigma r}}\biggr)^2 \mathrm{d}r\mathrm{d}t}_{=: K_j(\delta)} . \label{eq:IJKdef} \end{align}
 Therefore, to complete the proof we need to study the ($L$ dependent) terms $I(\delta)$, $J_\sigma(\delta)$ and $K(\delta)$. Physically, $I$ is a correction coming from the expansion of the Fermi momentum, $J$ comes from the density-dependence of $f$ (and is already present for the LDA), and $K$ comes from the dependence of $f$ on the norm of the density gradient (and is absent for the LDA). Let us start with  $I(\delta)$. In this case, we first note that 
\begin{align*}\frac{p_L^3}{p_F^3}-1 = \frac{3\pi |\partial Q^D|}{8 p_F |Q^D|}\frac{1}{L} + \mathcal{O}(L^{-\frac{35}{23}+\epsilon})
\end{align*}
by Lemma \ref{latticeeffect}. Hence by scaling out the $L$ in $I(\delta)$ and making the following observations: (i) $f^L_1(r,t)$ is bounded in $Q_{\delta,L}$ by a constant depending on $\delta$, but independent of $L$ (see \eqref{eq:fkest}) , (ii) $\lim_{L\ra \infty} f^L_1(Lr,t) = \partial_1 f(\bar{\rho},0)$ for a.e. $(r,t)\in Q^D_\frac{1}{2} \times [0,1]$, and (iii) $\lim_{L\ra \infty} \chi_{Q_{L,\delta}}(Lr) = \chi_{Q^D_\frac12}(r)$ in $L^1(\R^3)$, we conclude that
 \begin{align}
     I(\delta) =  \frac{3\pi \bar{\rho} |\partial Q^D|}{8 p_F} \partial_1 f(\bar{\rho},0) L^2 + \smallO(L^2) = c_{FM}(\bar{\rho})|\partial Q^D| L^2 + \smallO(L^2). \label{eq:Iest0}
 \end{align}
 We consider next the terms $J_\sigma$ with $trace(\sigma) \leq -1$. In this case, note that by the decay of $h$ (see Lemma \ref{fouriertransformball}), there exists some $j \in \{1,2,3\}$, such that $|h(p_L\norm{r-\sigma r})| \lesssim (1+\pi_j r)^{-2}$. Hence by estimate \eqref{eq:fkest} we find that $J_\sigma \lesssim_\delta L \log L$, and therefore, we just need to worry about the terms $K_j$ and $J_\sigma$ with $trace (\sigma) = 1$. For simplicity let us label the reflection $\sigma$ with $trace(\sigma) = 1$ and $\sigma_{jj}=-1$ by $\sigma_j$. Now note that, by the decay of $h$ and estimate \eqref{eq:fkest}, we have
\begin{align*}
    |f_1^L(r,t)h(2p_L r_j) + f_2^{L}(r,t)\frac{|\partial_{r_j} \rho^\textnormal{ctm}(r)|^2}{|\nabla \rho^\textnormal{ctm}(r)|}| \lesssim \frac{C(\delta)}{(1+r_j)^2} .
\end{align*}
As a consequence, up to an error bounded by $C(\delta) L$, one can change the domain of integration of $K_j(\delta)+J_{\sigma_j}(\delta)$ from $Q_{L,\delta}$ to 
\begin{align*}
    Q_{L,\delta}^j \coloneqq \{ r : r_j \in (0,\infty), \pi_j r \in \pi_jQ^D_{\frac{L}{2}}, \quad \min_{j\leq 3}|r_j| \geq \delta , \mbox{ and } \min_{j\leq3} |\pi_j r| > R(\delta) \}.
\end{align*}
In summary, we have
\begin{align}
    K_j(\delta)+J_{\sigma_j}(\delta)&= \frac{8\bar{\rho}p_L^3}{p_F^3}\int_0^1 \int_{\R^3} \chi_{Q_{L,\delta}^j}(r)\biggr(\frac{ \bar{\rho} p_L^3}{p_F^3}\frac{f_2^L(r,t)}{|\nabla \rho^\textnormal{ctm}_L(r)|} \biggr(\sum_{\sigma_{jj}=-1} \frac{p_L \det \sigma\dot{h}(p_L\norm{r-\sigma r})4r_j}{\norm{(1-\sigma)(r_j,L\pi_jr)}} \biggr)^2  \nonumber \\
    &+f_1^L(r,t)(-h(2p_L r_j))\biggr)\mathrm{d}r \mathrm{d} t  + \mathcal{O}(C(\delta) L) . \label{eq:mittelest}
\end{align}
Now note that by the decay of $h$, for a.e. $(r_j, \pi_j r) \in [0,\infty)\times \pi_j(Q^D_{\frac{1}{2}})$, the following holds:
\begin{align*} 
    &\lim_{L\ra \infty} \rho^\textnormal{ctm}_L(r_j, L\pi_j r)-\bar{\rho} = -\bar{\rho} h(2p_F r_j) ,\\
    &\lim_{L\ra \infty}|\nabla \rho^\textnormal{ctm}_L(r_j,L\pi_jr)|= 2\bar{\rho} p_F |\dot{h}(2p_F r_j)| ,\\
    %& h(2p_L r_j) \ra h(2p_F r_j)\quad \mbox{ and } \quad \sum_{\sigma \neq id} g^{\sigma}_j(p_Lr_j, p_L L \pi_j r) \ra - 2p_F \dot{h}(2p_F r_j)
    &\lim_{L\ra \infty} \sum_{\sigma_{jj}=-1} \det \sigma \bar{\rho}p_L\dot{h}\bigr(p_L \norm{(1-\sigma)(r_j,L\pi_jr)}\bigr) \frac{4r_j}{\norm{(1-\sigma_j)(r_j,L\pi_j r_j)}} = -2\bar{\rho} p_F \dot{h}(2p_F r_j) , \\
    &\lim_{L\ra \infty} \chi_{Q_{L,\delta}^j}(r_j, L\pi_jr) = \chi_{(\delta,\infty)}(r_j)\chi_{\pi_j(Q^D_{\frac12})}(\pi_j r) .
\end{align*}
Therefore, by scaling out the $L$ in the variables $\pi_jr$ in $K_j(\delta)$ and noting that the integrand in equation \eqref{eq:mittelest} is bounded by $C(\delta) \chi_{Q_{L,\delta}}(r)(1+|r_j|)^{-2}$ (by the decay of $h$), we conclude from dominated convergence in $Q^j_{L,\delta}$ and the continuity of $\nabla f$ in $Q^j_{L,\delta}$ that
\begin{align}
    K_j(\delta)+J_{\sigma_j}(\delta)&=2 \biggr(\prod_{\ell\neq j} d_\ell\biggr) L^2 \int_\delta^\infty \! \int_0^1\!\biggr(\partial_2 f\bigr(\bar{\rho}(1-th(2p_F r_j)),t2\bar{\rho}p_F|\dot{h}(2p_F r_j)|\bigr)2\bar{\rho}p_F|\dot{h}(2p_F r_j)| \nonumber \\
    &- \partial_1 f\bigr(\bar{\rho}(1-th(2p_F r_j)),t2\bar{\rho}p_F|\dot{h}(2p_F r_j)|\bigr)\bar{\rho} h(2p_F r_j)\biggr)\mathrm{d}t\mathrm{d}r_j \nonumber + \mathcal{O}(\delta L^2)+ \smallO_\delta(L^2)\nonumber \\
    &= \frac{\prod_{\ell\neq j} d_\ell}{p_F} L^2 \int_{2p_F \delta}^\infty f\bigr(\bar{\rho}(1-h(r_j),2\bar{\rho}p_F |\dot{h}(r_j)|\bigr)-f(\bar{\rho},0) \mathrm{d}r_j +\mathcal{O}(\delta L^2)+\smallO_\delta(L^2) , \label{eq:JKest}
\end{align}
where $\smallO_\delta$ emphasizes that the bounds may depend on $\delta$.
As a consequence, the proof follows from equations \eqref{eq:ctmsmallo},\eqref{eq:IJKdef},\eqref{eq:Iest0}, and \eqref{eq:JKest} by first taking the limit $L\ra \infty$ and then $\delta \ra 0$.
\end{proof}

\begin{remark*}[on Theorem~\ref{generalasymp}]
\begin{itemize} 
\item Note that $c_{FM}(\bar{\rho})$ above depends on $\bar{\rho}$ which is in contrast with the constant $c_{FM}$ defined in Theorem \ref{exactexchangethm} below. The reason is that for the ``physically" relevant cases (see the discussion before Theorem \ref{mainthm}), we have $c_{FM}(\bar{\rho}) = - c_{FM} \bar{\rho}$, where $c_{FM} = \frac{3}{8}$ is precisely the value defined there (see Corollary \ref{asymplda} below).
\item The same arguments can be used for semi-local functionals of higher order derivatives by considering generalized variables like $(\rho, \nabla \rho, ..., \partial^\alpha \rho)$ and applying estimate \eqref{centralest}. As we are only interested in LDA and GGAs for the moment, we leave the asymptotics of functionals for higher order derivatives for future works 
\item If $f$ is more regular and one knows the ratio at which the derivatives diverge when $(a,b) \ra 0$, one can further use the decay of $h$ to improve the remainder term from $\smallO(L^2)$ (for Dirichlet and Neumann cases) to, possibly, $\mathcal{O}(L^{\frac{34}{23}+\epsilon})$. \end{itemize} \end{remark*}

As straightforward corollaries of Theorem \ref{generalasymp}, we obtain the two-term asymptotic expansion for LDA and GGAs from Theorem \ref{mainthm}.

\begin{corollary}[Asymptotics of LDA] \label{asymplda}Let $\rho_{L}$ be the single-particle density of any determinantal ground state of the free $N$-electron gas in $Q_L^D$ (under our usual boundary conditions). Then in the thermodynamic limit we have
\begin{align*} E_x^{LDA}[\rho_L] = \begin{dcases}  -c_x \bar{\rho}^{4/3}|Q^D| L^3 + \mathcal{O}(L^{\frac{34}{23}+\epsilon}) , &\mbox{for periodic BCs,} \\
-c_x \bar{\rho}^{4/3}|Q^D| L^3 - c_{LDA}^{\textnormal{Dir}} \bar{\rho} |\partial Q^D| L^2 + \smallO(L^2), &\mbox{for Dirichlet BCs,} \\
-c_x \bar{\rho}^{4/3}|Q^D| L^3 - c_{LDA}^{\textnormal{Neu}}\bar{\rho} |\partial Q^D| L^2 + \smallO(L^2) , &\mbox{for Neumann BCs,}
\end{dcases}
\end{align*}
where the constants are
\begin{align*}
    c_{LDA} = \begin{dcases} c_{FM} + \frac{3}{8\pi} \int_0^\infty (1-h(s))^{\frac{4}{3}}-1 \mathrm{d}s, &\mbox{for Dirichlet BCs,} \\
    -c_{FM} + \frac{3}{8\pi}\int_0^\infty (1+h(s))^{\frac{4}{3}}-1 \mathrm{d}s , &\mbox{for Neumann BCs,} \end{dcases}
\end{align*}
%\begin{align} &E_x^{LDA}[\rho^{\textnormal{Per}}_L] = -c_x \bar{\rho}^{4/3}|Q^D| L^3 + \mathcal{O}(L^{\frac{34}{23}+\epsilon}) \label{asympldaper} \\
%&E_x^{LDA}[\rho^{\textnormal{Dir}}_L] = -c_x \bar{\rho}^{4/3}|Q^D| L^3 -  \underbrace{\biggr( \frac{3}{8\pi} \int_0^\infty (1-h(s))^{\frac{4}{3}}-1 \mathrm{d}s + c_{FM}\biggr)}_{\coloneqq c_{LDA}^{\textnormal{Dir}}} \bar{\rho} |\partial Q^D| L^2 + \mathcal{O}(L^{\frac{34}{23}+\epsilon})  \label{asympldadir} \\
%&E_x^{LDA}[\rho^{\textnormal{Neu}}_L] = -c_x \bar{\rho}^{4/3}|Q^D|L^3-  \underbrace{\biggr(\frac{3}{8\pi}\int_0^\infty (1+h(s))^{\frac{4}{3}}-1 \mathrm{d}s  - c_{FM}\biggl)}_{\coloneqq c^{\textnormal{Neu}}_{LDA}} \bar{\rho} |\partial Q^D|L^2 +  \mathcal{O}(L^{\frac{34}{23}+\epsilon}) \label{asympldaneu} \end{align}
with $c_{FM} = \frac{3}{8}$. (Compare with the constant in Theorem \ref{exactexchangethm}.)
\end{corollary}

\begin{corollary}[Asymptotics for GGA]\label{asympgga}
Let $g^{GGA}(a,b) \in C^0([0,\infty)^2)\cap C^1((0,\infty)\times[0,\infty))$ such that $g^{GGA}(a,0) = 0$, for all $a \geq 0$. Moreover, let $\rho_{L}$ be the single-particle density of any determinantal ground state of the free $N$-electron gas in $Q_L^D$ (under our usual boundary conditions). Then, for $\Delta E_x^{GGA}[\rho_L] = \int_{Q^D_L} g^{GGA}(\rho_L,|\nabla \rho_L|)$, in the thermodynamic limit we have
\begin{align*}
    \Delta E_x^{GGA}[\rho_L] = \begin{dcases} \mathcal{O}(L^{-\frac{34}{23}+\epsilon}) &\mbox{for periodic BCs,}\\
    c^{\textnormal{Dir}}_{GGA}(\bar{\rho}) |\partial Q^D| L^2 + \smallO(L^2) , &\mbox{for Dirichlet BCs,}\\
    c_{GGA}^{\textnormal{Neu}}(\bar{\rho})|\partial Q^D| L^2 + \smallO(L^2), &\mbox{for Neumann BCs,} \end{dcases}
\end{align*}
where the constants are given by
\begin{align*}
    c_{GGA}(\bar{\rho}) = \begin{dcases} \frac{1}{2p_F}\int_0^\infty g^{GGA}\biggr(\bar{\rho}(1-h(s)),2\bar{\rho}p_F|\dot{h}(s)|\biggl) \mathrm{d}s, &\mbox{ for Dirichlet BCs,} \\
    \frac{1}{2p_F} \int_0^\infty g^{GGA}\biggr(\bar{\rho}(1+h(s)), 2\bar{\rho} p_F| \dot{h}(s)|\biggl) \mathrm{d}s , &\mbox{ for Neumann BCs.} \end{dcases} 
\end{align*}

%\begin{align}
%&\Delta E_x^{GGA}[\rho^{\textnormal{Per}}_L] = \mathcal{O}(L^{\frac{34}{23}+\epsilon}) \label{asympggaper} \\
%&\Delta E_x^{GGA}[\rho^{\textnormal{Dir}}_L]= \underbrace{\frac{1}{2p_F}\int_0^\infty g^{GGA}\biggr(\bar{\rho}(1-h(s)),2\bar{\rho}p_F|\dot{h}(s)|\biggl) \mathrm{d}s}_{ =  c_{GGA}^{\textnormal{Dir}}(\bar{\rho})} |\partial Q^D|L^2 + \smallO(L^2) \label{asympggadir} \\
%&\Delta E_x^{GGA}[\rho^{\textnormal{Neu}}_L] = \underbrace{\frac{1}{2p_F} \int_0^\infty g^{GGA}\biggr(\bar{\rho}(1+h(s)), 2\bar{\rho} p_F| \dot{h}(s)|\biggl) \mathrm{d}s}_{ = c_{GGA}^{\textnormal{Neu}}(\bar{\rho})} |\partial Q^D| L^2 + \smallO(L^2)\label{asympgganeu} \end{align}
%where $\Delta E_x^{GGA}[\rho_L] = \int_{Q^D_L} g^{GGA}(\rho_L,|\nabla \rho_L|) \mathrm{d}r$.
\end{corollary}

\begin{proof}[Proof of Corollary \ref{asymplda} and \ref{asympgga}] In the first corollary, just apply Theorem \ref{generalasymp} to $f(\rho) = c_x \rho^{4/3}$. For the second one, set $f(\rho,|\nabla \rho|) = g^{GGA}(\rho,|\nabla \rho|)$ and note that $f(a,0) = 0$ for all $a \geq 0$ implies that $\partial_1 f(a,0) = 0$ for all $a >0$. 
\end{proof}

\subsection{Exact exchange}
Now we turn to the asymptotic expansion of the exact exchange. The analysis here is different as we are dealing with a functional of the density matrix, not the density.

\begin{theorem}[Asymptotics of exact exchange]
\label{exactexchangethm} Let $\bar{\rho}= \frac{N}{|Q^D|L^3} = constant$, $\Psi_L$ be a determinantal ground state of the free $N$-electron gas in $Q_L^D$. Then, we have
\begin{align*} E_x[\Psi_{L}] = \begin{dcases} -c_x\bar{\rho}^{\frac{4}{3}}|Q^D|L^3 + c_{FS} \bar{\rho} |\partial Q^D|L^2 + \mathcal{O}(L^{\frac{45}{23}+\epsilon}), &\mbox{for periodic BCs,} \\
-c_x \bar{\rho}^{\frac{4}{3}}|Q^D| L^3 - (c_{BL}^{\textnormal{Dir}}+c_{FM}-c_{FS})\bar{\rho}|\partial Q^D| L^2 + \mathcal{O}(L^{\frac{45}{23}+\epsilon}), &\mbox{for Dirichlet BCs,} \\
-c_x \bar{\rho}^{\frac{4}{3}}|Q^D| L^3 -(c_{BL}^{\textnormal{Neu}}-c_{FM}-c_{FS})\bar{\rho}|\partial Q^D| L^2 + \mathcal{O}(L^{\frac{45}{23}+\epsilon}) , &\mbox{for Neumann BCs,} \end{dcases} \end{align*}
%\begin{align} &E_x[\Psi^{\textnormal{Per}}_L] = -c_x\bar{\rho}^{\frac{4}{3}}|Q^D|L^3 + c_{FS} \bar{\rho} |\partial Q^D|L^2 + \mathcal{O}(L^{\frac{45}{23}+\epsilon}) \label{asympexactper} \\
%&E_x[\Psi^{\textnormal{Dir}}_L] = -c_x \bar{\rho}^{\frac{4}{3}}|Q^D| L^3 - (c_{BL}^{\textnormal{Dir}}+c_{FM}-c_{FS})\bar{\rho}|\partial Q^D| L^2 + \mathcal{O}(L^{\frac{45}{23}+\epsilon}) \label{asympexactdir} \\
%&E_x[\Psi^{\textnormal{Neu}}_L] = -c_x \bar{\rho}^{\frac{4}{3}}|Q^D| L^3 -(c_{BL}^{\textnormal{Neu}}-c_{FM}-c_{FS})\bar{\rho}|\partial Q^D| L^2 + \mathcal{O}(L^{\frac{45}{23}+\epsilon}) \label{asympexactneu} \end{align}
where the finite-size, Fermi momentum and boundary layer corrections are
\begin{align*} 
    c_{FS} = \frac{1}{8} , \quad c_{FM} = \frac{3}{8} ,\quad c_{BL}^{\textnormal{Dir}} = -\frac{\log 2}{4}\mbox{, and} \quad c_{BL}^{\textnormal{Neu}} = \frac{3\log 2}{4} .
\end{align*}
\end{theorem}

\begin{proof} As before, we prove the Dirichlet case in detail and outline the proof for the other two boundary conditions at the end. 

The first step is again to justify the use of the continuum density matrices \eqref{densitymatrixest}. For this, we use the identity \[ |\gamma_L(r,\tilde{r})|^2-|\gamma^\textnormal{ctm}_L(r,\tilde{r})|^2 = Re \{\overline{(\gamma_L(r,\tilde{r})-\gamma^\textnormal{ctm}_L(r,\tilde{r}))}(\gamma_N(r,\tilde{r})+\gamma^\textnormal{ctm}_N(r,\tilde{r}))\} \]
together with estimates \eqref{densitymatrixest0} and \eqref{densitymatrixdecay} from Theorem \ref{continuumapproxthm} to obtain
\begin{align*} \biggr| \int_{Q_L^D\times Q_L^D} \frac{|\gamma_L(r,\tilde{r})|^2-|\gamma^\textnormal{ctm}_L(r,\tilde{r})|^2}{|r-\tilde{r}|}\biggl| \mathrm{d}r\mathrm{d}\tilde{r} &\lesssim   L^{-\frac{70}{23}+\epsilon} \int_{Q_L^D\times Q_L^D}|r-\tilde{r}|^{-1}\mathrm{d}r\mathrm{d}\tilde{r} \\
& + L^{-\frac{35}{23}+\epsilon} \int_{Q_L^D\times Q_L^D} \frac{(1+|r-\tilde{r}|_{D,2L})^{-2}}{|r-\tilde{r}|}\mathrm{d}r\mathrm{d}\tilde{r} \\
&\lesssim L^{\frac{45}{23}+\epsilon} + L^{\frac{34}{23}+\epsilon} \log L . \end{align*}
Therefore, by the continuum formulas we have
\begin{align} E_x[\rho_L] \approx -\sum_{\tau,\sigma \in G} \det(\sigma \tau) \frac{\bar{\rho}^2 p_L^6}{4p_F^6} \underbrace{\int_{Q_L^D\times Q_L^D} \frac{h(p_L|r-\sigma \tilde{r}|_{D,2L})h(p_L|r-\tau \tilde{r}|_{D,2L})}{|r-\tilde{r}|}\mathrm{d}r\mathrm{d}\tilde{r}}_{\coloneqq J_{\sigma,\tau}(L)}  , \label{energyest1} \end{align}
where $\approx$ will be used throughout this proof to denote equality up to errors included in the remainder of Theorem \ref{exactexchangethm}. Now, to estimate the terms $J_{\sigma,\tau}$ we start with the following lemma.
\begin{lemma} If $trace(\sigma) \leq -1$ or $trace(\tau) \leq -1$, then $J_{\sigma,\tau}(L) \lesssim L (\log L)^2$. Furthermore, if $trace(\sigma) = 1 = trace(\tau)$ and $\sigma \neq \tau$, we have $J_{\sigma,\tau}(L) \lesssim L (\log L)^4$. \end{lemma}
\begin{proof}
First observe that since $r_j,\tilde{r}_j \in [0,d_jL]$, one has
\begin{align*} |r_j+\tilde{r}_j \mod 2d_j L| = \min\{r_j+\tilde{r}_j,2d_jL-r_j-\tilde{r}_j\} \geq |r_j-\tilde{r}_j| .\end{align*}
In particular, by the decay of $h$, we see that
\begin{align} |h(p_L|r-\sigma \tilde{r}|_{D,2L})| \lesssim \frac{1}{(1+|r-\sigma \tilde{r}|_{D,2L})^2} \lesssim \frac{1}{(1+|r-\tilde{r}|)^2}, \label{lemmaest1} \end{align}
for any $\sigma \in G$. On the other hand, if $trace(\sigma) \leq -1$, there exists $j\in \{1,2,3\}$ such that 
\begin{align*} |r-\sigma \tilde{r}|\geq \min_{\substack{p\in \pi_j(D(2L\Z)^3) \\  p \in  \pi_j(Q^D_{2L})}}|\pi_j(r+\tilde{r})-p| ,\end{align*}
and therefore,
\begin{align} |h(p_L|r-\sigma \tilde{r}|_{D,2L})| \lesssim\sum_{\substack{p\in \pi_j(D(2L\Z)^3) \\  p \in  \pi_j(Q^D_{2L})}} \frac{1}{(1+|\pi_j(r+\tilde{r})-p|)^2} .\label{lemmaest2} \end{align}
As a result, assuming that $trace(\sigma) \leq -1$, we can see from \eqref{lemmaest1} and \eqref{lemmaest2} that
\begin{align*} J_{\sigma,\tau}(L) \lesssim \int_{Q^D_L\times Q^D_L}\sum_{\substack{p\in \pi_j(D(2L\Z)^3) \\  p \in  \pi_j(Q^D_{2L})}} \frac{1}{(1+|\pi_j(r+\tilde{r})-p|)^2} \frac{1}{(1+|r-\tilde{r}|)^2}\frac{1}{|r-\tilde{r}|} \mathrm{d}r\mathrm{d}\tilde{r} \lesssim L (\log L)^2 . \end{align*}  
For the terms $J_{\sigma,\tau}$ with $trace(\sigma) =1 = trace(\tau)$ and $\sigma \neq \tau$, note that there exists $j\neq \ell \in \{1,2,3\}$ such that \begin{align*} |h(p_L |r-\sigma \tilde{r}|)|\lesssim \sum_{\substack{p_j \in \{0,2d_jL\}}}\frac{1}{(1+|r_j+\tilde{r}_j-p_j|)(1+|\pi_j(r-\tilde{r}|)} , \\
|h(p_L|r-\tau \tilde{r}|)| \lesssim \sum_{\substack{p_\ell \in \{0,2d_\ell L\}}}\frac{1}{(1+|r_\ell+\tilde{r}_\ell-p_\ell|)(1+|\pi_\ell(r-\tilde{r})|)} .\end{align*} 
 The lemma thus follows by integrating the product of the estimates above against the Coulomb potential in the box $Q_L^D\times Q_L^D$. \end{proof}
 
  From the lemma above, it is enough to study $J_{\textnormal{id},\textnormal{id}}(L)$, $J_{\sigma,\textnormal{id}}(L)$ and $J_{\sigma,\sigma}(L)$ where $trace(\sigma) = 1$.
 
 We start with $J_{\textnormal{id},\textnormal{id}}(L)$. In this case, we first note that $|r-\tilde{r}|_{D,2L} = |r-\tilde{r}|$ for any $r,\tilde{r}\in Q^D_L$. In particular, by the change of variables $(w(r,\tilde{r}),\tilde{w}(r,\tilde{r})) = (r-\tilde{r},\tilde{r})$, we obtain
\begin{align*} J_{\textnormal{id},\textnormal{id}}(L) &= |Q^D|L^3 \int_{Q_L^D-Q_L^D} \frac{h(p_L|w|)^2}{|w|} \mathrm{d} w -L^2 \sum_{j=1}^3 \biggr(\prod_{\ell \neq j} d_\ell\biggr) \int_{Q^D_L-Q_L^D} \frac{h(p_L|w|)^2|w_j|}{|w|}\mathrm{d} w \\
&+ \int_{Q^D_L-Q^D_L}\frac{h(p_L|w|)^2}{|w|} \biggr(L\sum_{j=1}^3 d_j  \prod_{\ell \neq j} |w_\ell| - \prod_{j=1}^3 |w_j|\biggr)\mathrm{d} w.
 \end{align*}
By the decay of $h$ and a simple estimate, up to an error $\lesssim L\log L$ the integral on the first two terms can be taken over the whole of $\R^3$, and the third integral can be neglected. In addition, by scaling out $p_L$, using spherical coordinates, and recalling that $|\partial Q^D| = 2(d_1d_2+d_1d_3+d_2d_3)$, we have
\begin{align} J_{\textnormal{id},\textnormal{id}}(L)  \approx \frac{4\pi|Q^D|L^3}{p_L^2} \underbrace{\int_{0}^\infty h(r)^2 r \mathrm{d} r}_{\coloneqq I_0} - \frac{\pi|\partial Q^D| L^2}{	p_L^3}  \underbrace{\int_{0}^\infty h(r)^2r^2\mathrm{d} r}_{\coloneqq I_1} . \label{J00est} \end{align}

Next, for the term $J_{\sigma,\sigma}$ we assume without loss of generality that $\sigma_{11} = -1$. Then, by invariance of the integrand under the reflections $r_i-\tilde{r}_i \mapsto \tilde{r}_i-r_i$, the change of variables $(w(r,\tilde{r}),\tilde{w}(r,\tilde{r})) = (r-\tilde{r},r_1+\tilde{r}_1,\tilde{r}_2,\tilde{r}_3)$ (notice $\frac{\mathrm{d}w\,\mathrm{d}\tilde{w}}{2} = \mathrm{d}r\,\mathrm{d}\tilde{r}$), and the decay of $h$, we have
\begin{align*} J_{\sigma,\sigma}(L) &= \int_{Q^D_L-Q^D_L}\int_{|w_1|}^{2d_1L-|w_1|} \frac{h(p_L|(\tilde{w}_1 \mod 2d_1L,\pi_1 w)|)^2\prod_{\ell \neq 1}(d_\ell L - |w_\ell|)}{|w|}\frac{\mathrm{d}\tilde{w}_1 \mathrm{d}w}{2}  \nonumber \\
&= 8d_2d_3 L^2 \int_{Q^D_L}\int_{w_1}^{d_1L} \frac{h(p_L|(\tilde{w}_1,\pi_1 w)|)^2}{|w|}\mathrm{d}\tilde{w}_1 \mathrm{d} w  + \mathcal{O}(L\log L) \nonumber \\
&= 8d_2d_3 L^2 \int_{Q^D_L}h(p_L|(\tilde{w}_1,\pi_1 w)|)^2\biggr(\int_{0}^{\tilde{w}_1}\frac{1}{|w|} \mathrm{d}w_1 \biggr)\mathrm{d}\tilde{w}_1 \mathrm{d}\pi_1 w + \mathcal{O}(L\log L) \end{align*}
(where we inverted the order of integration between $\tilde{w}_1$ and $w_1$ in the last step). In addition, since $\int_0^{|\tilde{w}_1} \frac{1}{|w|} \mathrm{d}w_1 = \frac{1}{2}\log\bigr(\frac{\tilde{w}_1+|(\tilde{w}_1,\pi_1 w)|}{|(\tilde{w}_1,\pi_1 w)|-\tilde{w}_1}\bigr) \lesssim \frac{|\tilde{w}_1|}{|\pi_1 w|}$, by the decay of $h$ one can see that, up to an error $\lesssim L$, we can replace the domain of the outer integration by $\R_+^3$. Hence, by scaling out $p_L$, changing to spherical coordinates, and recalling the definition of $I_1$ in \eqref{J00est}, we find that
\begin{align} J_{\sigma,\sigma}(L) &\approx \frac{2\pi d_2d_3}{p_L^3} L^2 \int_0^\infty h(r)^2 r^2\mathrm{d} r  \underbrace{\int_0^{\frac{\pi}{2}} \log\biggr(\frac{1+\cos \theta}{1-\cos \theta}\biggr)\sin \theta \mathrm{d}\theta}_{\mathclap{=(\cos \theta-1)\log(1-\cos\theta) -(1+\cos \theta)\log(1+\cos \theta)}}  = \frac{4\pi \log(2)}{p_L^3} d_2d_3 I_1 L^2 . \label{J11est} \end{align}

At last, for $J_{\textnormal{id},\sigma}$ (again assuming that $\sigma_{11}=-1$ without loss of generality), by the change of variables $(w(r,\tilde{r}),\tilde{w}(r,\tilde{r})) = (r-\tilde{r},r_1+\tilde{r}_1,\tilde{r}_2,\tilde{r}_3)$ and the same arguments as before we conclude that
\begin{align*} J_{\textnormal{id},\sigma} &= \int_{Q^D_L-Q^D_L} \int_{|w_1|}^{d_1L}\frac{h(p_L|(\tilde{w}_1,\pi_1 w)|)h(p_L|w|)\prod_{\ell \neq 1}(d_\ell L - |w_\ell|)}{|w|}\mathrm{d}\tilde{w}_1\mathrm{d} w  \\
&= 8d_2d_3 L^2 \int_{\R_+^3}\int_{|w_1|}^{\infty}\frac{h(p_L|(\tilde{w}_1,\pi_1 w)|)h(p_L|w|)}{|w|}\mathrm{d}\tilde{w}_1\mathrm{d} w + \mathcal{O}(L\log L) .
\end{align*}
Hence, scaling out $p_L$ and using spherical coordinates for $w$, we have
\begin{align} J_{\textnormal{id},\sigma}(L) \approx \frac{4\pi d_2d_3}{p_L^3} L^2 \underbrace{\int_0^\infty \int_0^{\frac{\pi}{2}} \int_{r\cos \theta}^\infty h(\sqrt{(r\sin\theta)^2+\tilde{w}_1^2})h(r) r \sin \theta \mathrm{d}\tilde{w}_1 \mathrm{d}\theta \mathrm{d} r}_{\coloneqq I_2}. \label{J10est} \end{align}
Thus by plugging \eqref{J00est},\eqref{J11est} and \eqref{J10est} into \eqref{energyest1}, we have
\begin{align*} E_x[\rho_L] &\approx - \frac{\bar{\rho}^2 p_L^6}{ 4 p_F^6} \biggr( J_{\textnormal{id},\textnormal{id}}(L) + \sum_{j=1}^3 J_{\sigma_j,\sigma_j}(L) - \sum_{j=1}^3 \bigr(J_{\textnormal{id},\sigma_j}(L) + J_{\sigma_j,\textnormal{id}}(L)\bigr) \biggr) \\
&\approx -\frac{\pi \bar{\rho}^2 p_L^4}{p_F^6}I_0|Q^D| L^3 - \biggr(-\frac{\pi\bar{\rho}^2 p_L^3}{4 p_F^6}I_1 +\frac{\pi \log 2\bar{\rho}^2 p_L^3}{2p_F^6}I_1- \frac{\pi \bar{\rho}^2p_L^3}{p_F^6}I_2\biggr)|\partial Q^D| L^2  .\end{align*}
As a result, using Lemma \ref{latticeeffect} to replace $p_L$ by $p_F$ plus correction, we conclude that
\begin{align*} E_x[\rho_L] \approx  -\underbrace{\frac{\pi}{(3\pi^2)^{\frac{2}{3}}} I_0}_{c_x} \bar{\rho}^{\frac{4}{3}}|Q^D|L^3 - \biggr(\underbrace{\frac{1}{6}I_0}_{c_{FM}} - \underbrace{\frac{1}{12 \pi} I_1}_{c_{FS}} + \underbrace{\frac{\log 2}{6\pi}I_1- \frac{1}{3\pi}I_2}_{c_{BL}}\biggr)\bar{\rho} |\partial Q^D| L^2 . \end{align*}
The proof is completed by using the values of the integrals computed in the next lemma.
\begin{lemma}\label{constlemma} $I_0 = \frac{9}{4}$, $I_1 = \frac{3\pi}{2}$ and $I_2 = \frac{3\pi \log 2}{2}$.
\end{lemma}
\begin{proof} For $I_0$, we can use the identity $\int \frac{(\sin s - s \cos s)^2}{s^5} \mathrm{d} s = \frac{-2s^2+2 s \sin(2s) + \cos(2s)-1}{s^4}$ (see e.g., \cite[Lemma 6.1]{Friesecke1997} for an elegant evaluation) to obtain
\begin{align*} I_1 = \int_0^\infty 9 \frac{(\sin s - s \cos s)^2}{s^5} \mathrm{d}s = -9 \lim_{s\ra 0} \frac{-2s^2+2 s \sin(2s) + \cos(2s)-1}{s^4} = \frac{9}{4} , \end{align*}
where the limit can be computed by L'Hôpital's rule. Next, note that for any $a>0$, by Plancherel's theorem, we have
\begin{align} \int_0^\infty h(r)h(ar)\, r^2\mathrm{d}r &= \frac{1}{4\pi |B_1|^2} \int_{\R^3} \widehat{\chi}_{B_1}(k)\widehat{\chi}_{B_1}(ak) \mathrm{d} k \nonumber \\
&= \frac{(2\pi)^3}{4\pi |B_1|^2 a^3} \int_{\R^3} \chi_{B_1}(k) \chi_{B_1}\biggr(\frac{k}{a}\biggr) \mathrm{d}k \nonumber \\
 &= \frac{3\pi}{2 \max\{a,1\}^3} . \label{ballintegral} \end{align}
In particular, the value of $I_1$ follows by setting $a=1$. For $I_2$, first note that by using the inverse of polar coordinates $(r,\theta) = (\sqrt{x^2+y^2}, \arctan(y/x))$, one has
\begin{align*}
I_2 = \int_{\R_+^2} \int_{x}^\infty h\bigr(\sqrt{y^2+\tilde{w}_1^2}\bigr)h\bigr(\sqrt{x^2+y^2}\bigr)\frac{y}{\sqrt{x^2+y^2}} \mathrm{d}\tilde{w}_1\mathrm{d}x\mathrm{d}y .\end{align*}
Next, set $\mathcal{R} \coloneqq \{(x,y,\tilde{w}_1)\in \R^3 : x>0,y>0,\tilde{w}_1 >x\}$ and consider the change of coordinates 
\begin{align*} &T:\biggr\{(s,\phi,v)\in \R^3: s>0, 0<\phi<\frac{\pi}{2}, 0<v<\log \biggr(\frac{\sin\phi+1}{\cos\phi}\biggr)\biggr\} \ra \mathcal{R} \\
&(s,\phi,v) \mapsto (x,y,\tilde{w}_1) = T(s,\phi,v) =  (s \cos \phi \sinh v, s\sqrt{\sin^2 \phi - \cos^2 \phi \sinh^2 v},s\cos \phi \cosh v), \end{align*}
for which
\begin{align*} \det D T(s,\phi,v) = s^2 \frac{\sin (\phi) \cos (\phi)}{\sqrt{\sin^2 \phi - \cos^2 \phi \sinh^2 v}} =  s^2 \cos \phi\frac{\sqrt{x^2+y^2}}{y} .\end{align*}
Then, by \eqref{ballintegral} and the substituion $\tau=\sin \phi$, we conclude that 
\begin{align*} 
    I_2 &= \int_0^{\frac{\pi}{2}} \cos \phi \, \log\biggr(\frac{\sin \phi+1}{\cos \phi}\biggr)\underbrace{\int_{0}^\infty h(s)h(s\sin \phi)s^2 \mathrm{d} s}_{= \frac{3\pi}{2}}\mathrm{d}\phi \\
    &=  \frac{3\pi}{2}\int_0^1 \log \biggr(\frac{1+\tau}{\sqrt{1-\tau^2}} \biggr)\mathrm{d} \tau = \frac{3\pi}{2}\int_0^1 \log \biggr(\frac{\sqrt{1+\tau}}{\sqrt{1-\tau}} \biggr)\mathrm{d} \tau = \frac{3\pi \log 2}{2}, 
\end{align*}
which completes the proof of Theorem \ref{exactexchangethm} for the Dirichlet case. For the Neumann case the same arguments work with the proper change of signs (due to the missing $\det \sigma$ in \eqref{densitymatrix}). \end{proof}

For the periodic case, we first note that by the decay of $h$,
\begin{align*}
    |h(p_L|r-r'|_{L,D}|)^2 - \sum_{\substack{p \in DL \Z^3 \\ |D^{-1} p|_{\max} \leq L}} h(p_L|r-r'-p|)^2 \lesssim L^{-4} .
\end{align*}
Moreover, one can now show that
\begin{align*}
    \int_{Q^D_L\times Q^D_L} \frac{h(p_L|r-\tilde{r}-p|)^2}{|r-\tilde{r}|} \mathrm{d}r\mathrm{d}\tilde{r}  \lesssim L \log L \quad \mbox{for any $p \in DL\Z^3 \setminus \{0\}$,}
\end{align*}
which implies that $E_x[\Psi_L] \approx -\frac{\bar{\rho}^2 p_L^6}{4p_F^6} J_{id,id}(L)$. The result now follows from Lemma~\ref{latticeeffect}, estimate~\eqref{J00est} and Lemma~\ref{constlemma}.
\end{proof}

\section{Kinetic energy } \label{sec:kinetic}
In this section we use Lemma \ref{central} to compute the asymptotic expansion of the kinetic energy. 
\begin{theorem}[Asymptotics of kinetic energy] \label{kineticthm} Let $D$ fixed, $\bar{\rho} = \frac{N}{|Q^D|L^3} = constant$ and $\Psi_{N,L}$ be any ground state of the free $N$-electron gas in the box $Q^D_L$ under our usual boundary conditions. Let $T$ be the kinetic energy functional defined in \eqref{kinen}. Then, we have 
\begin{align*}  T[\Psi_{N,L}] = \begin{dcases} 
c_{TF} \bar{\rho}^{5/3} |Q^D|L^3+ \mathcal{O}(L^{\frac{34}{23}+\epsilon})  , &\mbox{for periodic BCs,} \\ 
c_{TF} \bar{\rho}^{5/3} |Q^D|L^3 + c_{K} \bar{\rho}^{4/3} |\partial Q^D|L^2 + \mathcal{O}(L^{\frac{34}{23}+\epsilon}) , &\mbox{for Dirichlet BCs,} \\
c_{TF} \bar{\rho}^{5/3}|Q^D| L^3 - c_{K} \bar{\rho}^{4/3} |\partial Q^D| L^2 + \mathcal{O}(L^{\frac{34}{23}+\epsilon}), &\mbox{for Neumann BCs,} \end{dcases} \end{align*}
where $c_{TF} = \frac{3}{10}(3\pi^2)^{2/3}$ is the Thomas-Fermi constant and $c_K = \frac{3 \pi}{32}(3\pi^2)^{1/3}$. 
\end{theorem}

\begin{proof}
For the Dirichlet case, just note that
\begin{align*} T[\Psi_{N,L}] &= \inner{\Psi_{N,L},-\frac{\Delta}{2}\Psi_{N,L}} = \!\!\!\sum_{k \in \N^3 \cap B^D_{R_N}} \!\!\! \frac{\pi^2 |D^{-1}k|^2}{L^2} - (N_+-N) \frac{\pi^2 R_N^2}{L^2} \\ 
&= \frac{\pi^2}{L^2}\frac{1}{8}\biggr(\sum_{k \in \Z^3 \cap B^D_{R_N}}\!\!\! |D^{-1} k|^2 - \sum_{j=1}^3 \sum_{\substack{k\in \Z^3 \cap B^D_{R_N} \\ k_j = 0}}|D^{-1} k|^2\biggr) + \mathcal{O}(L^{\frac{34}{23}+\epsilon}) . \end{align*} Moreover, by a simple estimate\footnote{In fact, by adapting the proof of \cite{Friesecke1997}, one can get a remainder of order  $\mathcal{O}(R^{2+\frac{2}{3}})$ in estimate \eqref{variation}.} we find that
\begin{align} \sum_{\substack{k \in B^D_{R_N}\cap \Z^3 \\ k_j = 0}} |D^{-1}k|^2 = \int_{\R^2}\chi_{\pi_j(B^D_{R_N})}(k)\sum_{\ell \neq j} (d_{\ell}^{-1}k_{\ell})^2 \mathrm{d}k + \mathcal{O}(R_N^{3}) = \frac{\pi}{4} R_N^4 \prod_{\ell \neq j} d_\ell  + \mathcal{O}(R_N^3). \label{variation} \end{align}
Thus using estimate \eqref{variation} and Lemma \ref{central}, we conclude that
\begin{align*} T[\Psi_{N,L}] = \frac{1}{10 \pi^2} p_{N,L}^5 |Q^D| L^3 -\frac{1}{32 \pi}p_{N,L}^4|\partial Q^D| L^2 + \mathcal{O}(L^{\frac{34}{23}+\epsilon}). \end{align*}
The result now follows from the asymptotics of $p_L$ in Lemma \ref{latticeeffect}. The Neumann and periodic cases are entirely analogous. 
\end{proof}

\begin{remark*} Note that we do not assume $\Psi_{N,L}$ to be a determinantal ground state as the kinetic energy is simply the ground state energy of the Laplacian and therefore unique (even if the ground state is not). \end{remark*}

\section{Numerical results and final discussion} \label{sec:numerics}

We now compare our asymptotic results to numerical values of different exchange functionals for the free 
electron gas with zero boundary conditions, for up to $30~\! 000$ electrons.  
Our numerical computations were carried out in Matlab. All energy functionals other than  exact exchange were evaluated by direct numerical integration of the exact formulas given in section~\ref{sec:electrongasformulas}. For exact exchange, accurate direct numerical evaluation of the expression \eqref{exchDM}, \eqref{densitymatrix} is impossible, because of the high-dimensionality of the domain of integration (6D) and the ${1}/{|r-r'|}$ singularity of the integrand. We tackled these obstructions by reducing the problem to the numerical computation of a small ($O(N^{1/3})$) number of one-dimensional integrals of smooth functions (see appendix~\ref{appendix:numerics} for a detailed description). Moreover, we focus here on the case of a cubic box $[0,L]^3$ and $\bar{\rho}=1$. 

To begin with, in Figure \ref{F:firstorderconvergence} we have plotted the exact exchange energy per unit volume, together with the theoretical one-term (just $c_x$) and two-term asymptotics ($c_x + c_{x,2}\cdot 6 L^{-1}$) from Theorem \ref{mainthm}. For comparison we have also included the LDA exchange energy per unit volume.
\begin{figure}[ht!]
    \centering
    \includegraphics[width=0.7\textwidth]{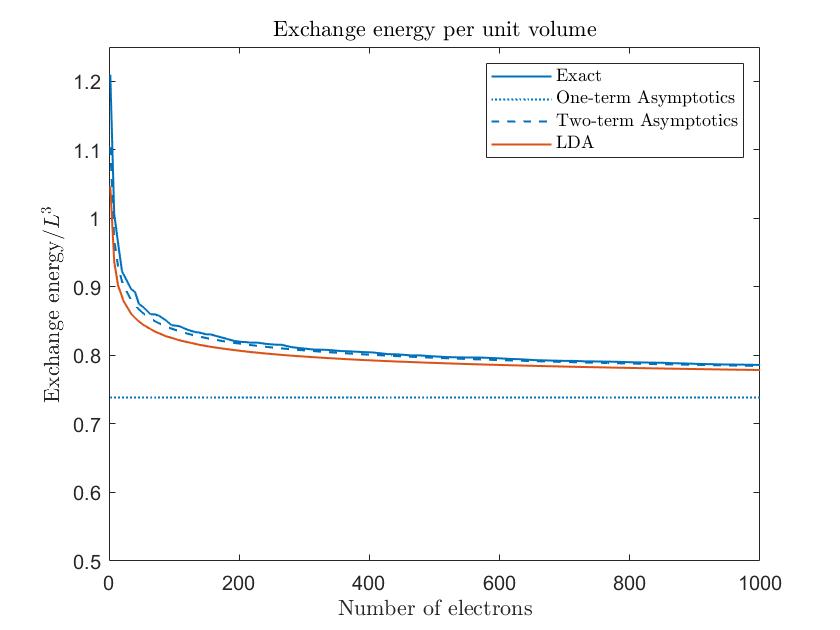}
    \caption{Exact exchange energy per unit volume of the free electron gas in a box with zero boundary conditions, compared with the LDA, one-term asymptotics (Dirac exchange constant) and two-term asymptotics (present work, Theorem \ref{mainthm}).}
    \label{F:firstorderconvergence}
\end{figure}
Note that even for small $N$,  the two-term asymptotics is a much better approximation than the one-term asymptotics, and also a better one than the more complicated LDA. Note that the latter requires integration of an inhomogeneous $N$-electron exchange energy density of the system.

Let us now look in more detail at the next-order contribution. Besides exact exchange and the LDA, we consider the widely used GGAs introduced by Becke in 1988 (B88) \cite{Becke1988} and  Perdew, Burke and Ernzerhof in 1996 (PBE) \cite{PBE1996}, and the modified version of PBE introduced by Perdew et al. in 2008 (PBEsol) \cite{PBEsol2008}. For convenience of the reader, we recall the expressions for these functionals here:
\begin{align} &g^{B88}(\rho,|\nabla \rho|) =  \frac{2^\frac13 \beta \bigr(|\nabla \rho|/\rho^\frac43\bigr)^2}{1 + 6 \beta 2^\frac13 \bigr(|\nabla \rho|/\rho^\frac43\bigr)  \sinh^{-1}\bigr(2^\frac13 |\nabla \rho|/\rho^\frac43\bigr)} \rho^\frac43 \label{b88} \\ \nonumber \\
&g^{PBE}(\rho,|\nabla \rho|) = c_x \frac{\mu \bigr(|\nabla \rho|/\rho^\frac43\bigr)^2}{4(3\pi^2)^\frac23 + \frac{\mu}{\kappa} \bigr(|\nabla \rho|/\rho^\frac43\bigr)^2} \rho^\frac43\label{pbe}
\end{align}
where $\sinh^{-1}$ is the inverse hyperbolic sine and the constants are $\beta = 0.0042$, $\kappa = 0.804$, and $\mu = 0.2195$. For PBEsol, one has the same expression as for PBE in  \eqref{pbe}, but with $\mu = 0.1235$.

\begin{figure}[ht!]
    \centering
    \includegraphics[scale=0.56]{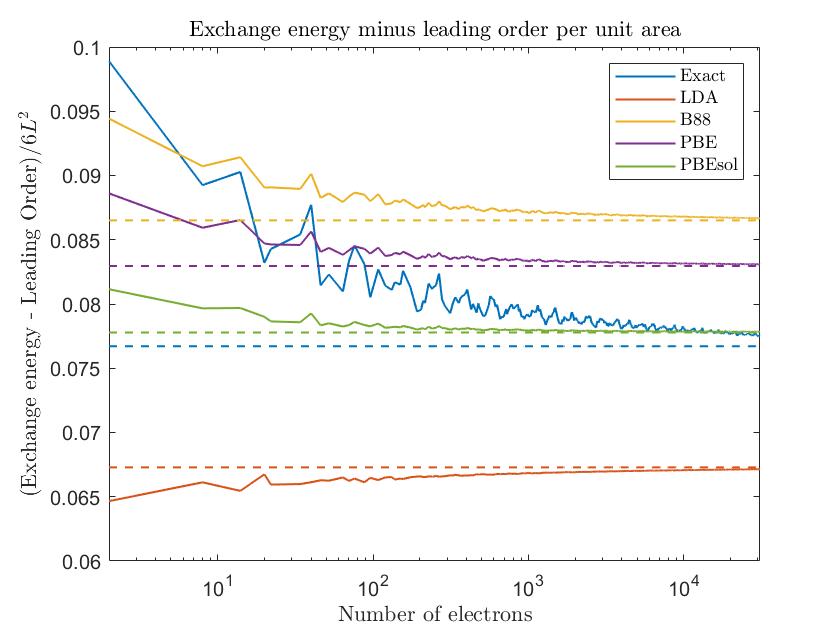}
    \caption{Exchange energy of the free electron gas in a box with zero boundary conditions minus leading order term, per unit boundary area, for various functionals. The number of
electrons per unit volume was normalized to 1. Solid lines: Numerical values. Dashed lines: asymptotic values (second order coefficients from Theorem \ref{mainthm}, present work).}
    \label{F:secondorderconvergence}
\end{figure}

Numerical evaluation of the exact one-dimensional integral expression for the GGA constant in Theorem~\ref{mainthm} gives the following values:
\begin{equation}  c^{\textnormal{Dir}}_{PBE} \approx 0.0157 ,\quad c^{\textnormal{Dir}}_{B88} \approx 0.0192, \quad c^{\textnormal{Dir}}_{PBEsol} \approx 0.0105.  \label{constantvalues} \end{equation}
To numerically verify the next-order asymptotics, we plotted in Figure \ref{F:secondorderconvergence} the graph of the energy functionals minus the leading order term divided by the boundary area $|\partial [0,L]^3|=6L^2$, together with the asymptotic values predicted by Theorem \ref{mainthm} and \eqref{constantvalues}. Precisely, since $\bar{\rho} =1$, the values for the dashed lines in Figure \ref{F:secondorderconvergence} are, respectively, $c^{\textnormal{Dir}}_{x,2}$, $c^{\textnormal{Dir}}_{\rm LDA}$, $c^{\textnormal{Dir}}_{\rm LDA}+c^{\textnormal{Dir}}_{B88}$, $c^{\textnormal{Dir}}_{\rm LDA}+c^{\textnormal{Dir}}_{PBE}$, $c_{\rm LDA}^{\textnormal{Dir}}+c^{\textnormal{Dir}}_{PBEsol}$.

Overall, there is a good match between numerics and asymptotics as $N$ gets large. More detailed observations are the following. 

\begin{itemize}
\item Asymptotically, the LDA underestimates the surface term by 12\%, whereas B88 and PBE overestimate it by 13\% respectively 8\%. 
\vspace*{-1mm}

\item Asymptotically, only PBEsol is much more accurate than the LDA, exhibiting an error of just  1.4\%. This should not come as a surprise to experts, as certain surface data (although not the ones considered here) entered into the choice of the parameters. Thus one may say that the present work provides an alternative theoretical justification of the PBEsol parameters. The price to pay is that PBEsol is the least accurate of the GGAs for very small $N$.
\vspace*{-1mm}

\item  B88 is the most accurate GGA for very small $N$. This is not unexpected given the fact that the parameter $\beta$ was fitted to data for the first few noble gas atoms. The price to pay is that B88 does not improve on the LDA beyond a few hundred electrons. 
\vspace*{-1mm}

\item In the regime of 20 to 100 electrons, which is certainly relevant in applications, particularly in chemistry, PBE fares best.
\vspace*{-1mm}

\item The slowest convergence to the asymptotic value, and the largest fluctuations, occur for exact exchange. Neither asymptotics up to second order nor any of the GGA functionals correctly reproduce these significant finite-$N$ fluctuations. 
%which show that the seemingly simple free electron gas with zero boundary conditions in a box is not that simple! 
Note that they would be captured exactly by the universal Hohenberg-Kohn functional.
\end{itemize}

In the context of our model system, the free electron gas in a box with zero boundary conditions,
our rigorous asymptotic results and the above observations illustrate
both the advances that have been made in the physics and chemistry literature in designing computationally simple exchange-correlation functionals, and the immense difficulties in improving on the current state of the art. For the latter, we hope that the new exact constraint on GGAs presented here (eq.~\eqref{exactconstr}) will in the future turn out to be useful.

\section*{Acknowledgments} We are grateful to Fernando Chamizo for a helpful email exchange regarding his results. We would also like to thank Mi-Song Dupuy for
useful comments on an early draft of this manuscript and helpful discussions regarding the numerical evaluation of exact exchange.

\appendix

\section{Numerical scheme for exact exchange} \label{appendix:numerics}

Here we detail our scheme for accurate and efficient evaluation of exact exchange for the free electron gas in a box. As already explained, the closed-form expression \eqref{exchDM}, \eqref{densitymatrix} cannot be evaluated directly by numerical integration, because of the high-dimensionality of the domain of integration (6D) and the Coulomb singularity of the integrand.

Recall that the eigenfunctions of the Laplacian with zero boundary conditions on $Q = [0,1]^3$ are given by the following expression:
\begin{align} \phi_k(x) =  \prod_{i=1}^3 \sqrt{2}\sin( \pi k_i x_i)  \; \; \; (k\in\N^3). \label{eigenfunction} \end{align}
Hence for closed shell $N$, the ground state $\Psi_{N,L}$ of the free $N$-electron gas in $Q_L=[0,L]^3$, $N/L^3=1$, is unique and the exact exchange energy is, by rescaling to the fixed domain $Q$,  
\begin{align} E_x[\Psi_{N,L}] &= -\frac{1}{L}\int_{Q\times Q} \frac{\left|\sum_{k \in \N^3, \, |k|\le R_N}\phi_k(x) \overline{\phi_k(y)}\right|^2}{|x-y|} dx \, dy \nonumber \\
 &=-\frac{1}{L}\sum_{\substack{k,\ell\in\N^3\\|k|,|\ell|\le R_N}} \underbrace{ \int_{Q\times Q} \frac{\prod\limits_{i=1}^3 4\sin(\pi k_i x_i) \sin(\pi \ell_i x_i)\sin(\pi k_i y_i) \sin(\pi \ell_i y_i)}{|x-y|} dx \, dy}_{=:I_{k,\ell}} 
 \label{A2} \end{align}
where $R_N$ is the Dirichlet Fermi radius defined in \eqref{fermiradiidef}. 

The starting point of our numerical scheme is a simple calculation which reduces the above six-dimensional integral to the three-dimensional integral of a separable function times the Coulomb potential over a finite region. For periodic boundary conditions such a reduction is trivial because the system is translation invariant, which implies that the exchange integrand depends only on the relative coordinate $z=x-y$; but the zero boundary condition breaks the translation invariance. Nevertheless the following holds:
\begin{lemma}\label{reductionlemma} For $k$, $\ell\in\N^3$, and $I_{k,\ell}$ as defined above,
\begin{align} I_{k,\ell} = 8 \int_{[0,1]^3} \frac{\prod_{i=1}^3 f_{k_i,\ell_i}(z_i)}{|z|} dz\label{A3} \end{align}
where for $a,b \in \N$, $f_{a,b}$ is defined as
\begin{align} f_{a,b}(\tau) & = \frac{1-\tau}{2}\Bigl(\cos(\pi(a+b)\tau) + \cos(\pi (a-b) \tau)\Bigr)
+ \biggr(\frac{1}{a}  + \frac{1}{b} - \frac{1}{(a+b)}\biggr)\frac{\sin(\pi(a+b)\tau)}{2\pi} 
\nonumber \\
& 
+ \begin{cases} \Bigl(\frac{1}{a} - \frac{1}{b} - \frac{1}{a-b}\Bigr)\frac{\sin(\pi(a-b) \tau)}{2\pi} & \mbox{ if } a\neq b, \\ \;\;\;\; \frac{1-\tau}{2} & 
\mbox{otherwise.} \end{cases} \label{ffunc}\end{align}
\end{lemma}
We remark that there is a well known alternative reduction of any 6D Coulomb integral of the form $\int u(x) \frac{1}{|x-y|}v(y)\, dx\, dy$ to a 3D integral over reciprocal space, by using the convolution theorem for the Fourier transform. But this leads to an integral over an unbounded domain, a stronger ($\sim 1/|k|^2$) singularity, and -- in our case -- a slow decay of the integrand, making the expression \eqref{A3} numerically much more  favourable.
\begin{proof} Using the identity 
$2\sin(A)\sin(B) = \cos(A-B) - \cos(A+B)$ and the change of variables $w = x+y, \, z = x-y$, and noting that the volume element becomes $\mathrm{d}x_i\mathrm{d}y_i = \frac{1}{2}\mathrm{d}w_i \mathrm{d}z_i$, 
$$
 I_{k,\ell} = \int_{[-1,1]^3} \frac{1}{|z|}
 \prod_{i=1}^3 \underbrace{ \left\{
 \frac12 \int_{|z_i|}^{2-|z_i|} \bigl( \cos\pi k_iz_i - \cos\pi k_i w_i\bigr) \bigl( \cos\pi\ell_iz_i - \cos\pi\ell_iw_i\bigr) dw_i 
 \right\} }_{=:f_{k_i,\ell_i}(z_i)} dz.
$$
The integral in the definition of $f_{a,b}(z_i)$ is elementary to evaluate by using the identity $2\cos A \cos B = cos(A+B) + \cos(A-B)$ for the term $\cos \pi a w_i \cos \pi b w_i$, yielding
\begin{align*} f_{a,b}(z_i) &  = (1-|z_i|)\cos(\pi a z_i)\cos(\pi b z_i) + \cos(\pi a z_i)\frac{\sin(\pi b |z_i|)}{\pi b}+ \cos(\pi b z_i) \frac{\sin(\pi a |z_i|)}{\pi a}  \nonumber \\
& - \frac{\sin(\pi(a+b)|z_i|)}{2\pi (a+b)}\;\; 
\begin{cases}
- \frac{\sin(\pi (a-b)|z_i|)}{2\pi(a-b)} &    \mbox{if } a\neq b, \\
+ \frac{1-|z_i|}{2} & \mbox{otherwise.}
\end{cases}
\end{align*}
 Next, we note that since $f_{a,b}(-z_i) = f_{a,b}(z_i)$ and the Coulomb potential $1/|z|$ is invariant under the transformations $z_i \mapsto -z_i$, we can replace $[-1,1]^3$ by $[0,1]^3$ by adding a factor of $8$ in front of the integral. The final expression for $f_{a,b}$ given in the lemma now follows from the trigonometric identities $2\cos A \cos B = \cos(A+B)+\cos(A-B)$ and $\alpha \sin A \cos B + \beta \cos A \sin B= \frac{\alpha+\beta}{2}\sin(A+B)+\frac{\alpha-\beta}{2}\sin(A-B)$.
\end{proof}

Since the integrand in \eqref{A3} is a separable function except for the Coulomb potential, the idea now is to also approximate the latter by separable functions, therefore reducing the problem to the computation of one dimensional integrals. Such an approximation is provided by recent advances in low-rank tensor approximation; more specifically, we use results of Hackbusch (\cite[Section 9.8.2]{hackbusch2012}). The  Coulomb potential can be very accurately approximated by a sum of weighted Gaussians, \begin{align} 
   \frac{1}{r} \approx \sum_{j=1}^M \omega_j e^{-\alpha_j r^2}.
\label{Gaussian}   
\end{align}
Plugging this approximation into equation \eqref{A3} and factorizing $e^{-\alpha_j|z|^2}=\prod_{i=1}^3 e^{-\alpha_jz_j^2}$, one obtains 
\begin{align} 
 I_{k,\ell} \approx  8 \sum_{j=1}^M \omega_j \prod_{i=1}^3 I_{k_i,\ell_i,j} \; \mbox{ with } \; I_{a,b,j} = \int_0^1 f_{a,b}(t) e^{-\alpha_j t^2} dt, \label{firstreduction} \end{align}
which reduces the 3D integral in \eqref{A3}  to one-dimensional integrals of analytic functions.\footnote{In fact, one could represent these integrals exactly in terms of the error function ${\tt erf}$ and the imaginary error function ${\tt erfi}$, but we do not use this fact on our scheme.}

To reduce the overall number of 1D integrals that must be computed, let us introduce, for $p\in\{0,1,...,2\lfloor{ R_N^{\textnormal{Dir}}} \rfloor \}$ and $j\in\{1,...,M\}$ (where $\lfloor \; \rfloor$ denotes the integer part), the auxiliary integrals 
\begin{align} C_{p,j} \coloneqq \int_0^1 \frac{1-t}{2} \cos(\pi p \, t) e^{-\alpha_j t^2} \mathrm{d} t \quad \mbox{ and } \quad S_{p,j} \coloneqq \int_0^1 \frac{\sin(\pi p \, t)}{2\pi} e^{-\alpha_j t^2} \mathrm{d} t.  \label{CandSintegrals} 
\end{align}
It follows from the explicit expression for $f_{a,b}$ in \eqref{ffunc} that
\begin{align}
 I_{a,b,j} & = C_{a+b,j} + C_{|a-b|,j} + \Bigl( \tfrac{1}{a} + \tfrac{1}{b} - \tfrac{1}{a+b}\Bigr) S_{a+b,j} \nonumber \\
 & \;\;\; + \begin{cases} \Bigl( \tfrac{1}{a}-\tfrac{1}{b}-\tfrac{1}{a-b}\Bigr) {\rm sign}(a-b) S_{|a-b|,j} & \mbox{ if } a\neq b \\
\; C_{0,j} & \mbox{ if }a=b. \end{cases} \label{numericalexpression} 
\end{align}
Thus in total
\begin{equation} \label{A...}
   E _x[\Psi_{N,L}] \approx - \frac{8}{L} 
   \sum_{\substack{k,\ell\in\N^3\\|k|,|\ell|\le R_N}} \sum_{j=1}^M \omega_j \prod_{i=1}^3 I_{k_i,\ell_i,j}
\end{equation}
with $I_{a,b,j}$ given by \eqref{CandSintegrals}--\eqref{numericalexpression}. 
In particular, as $R_N^{\textnormal{Dir}} \sim N^{1/3}$, calculating the exchange energy of the free $N$-electron gas reduces to the problem of evaluating $\mathcal{O}(N^{1/3} M)$ 
one-dimensional integrals of analytic functions on the interval $[0,1]$, and multiplying and summing them according to equation \eqref{A...}.

Next, let us discuss the choice of weights and exponents, and the error, in \eqref{Gaussian}. We used the values 
$\{\omega^H_j,\alpha^H_j\}_{j=1}^{51}$ given on Hackbusch's webpage \cite{Hackbuschwebsite} for the (approximately) best approximation of the Coulomb potential as the sum of $M=51$ Gaussians, which satisfy 
$
 \norm{1/r-v^H(r)}_{L^\infty([1,10^9])} \leq 10^{-9} 
$
where $v^H(r) = \sum_{j=1}^{51} \omega_j^H e^{-\alpha_j^H r^2}$. Moreover since we are interested in a good approximation of $1/r$ on the unit cube, we rescaled Hackbusch's parameters by setting
\begin{align*} \omega_j = \frac{\omega^H_j}{r_0}\quad \mbox{ and } \quad \alpha_j = \frac{\alpha_j^H}{r_0^2}, \end{align*}
which yields a pointwise error of  \begin{align} 
  \norm{v(r)-1/r}_{L^\infty([r_0,10^9 r_0])}\leq r_0^{-1}\times 10^{-9} \label{coulombLinftyest} 
\end{align} 
where $v(r) = \sum_{j=1}^{51} \omega_j e^{-\alpha_j r^2}$. In our numerical results we chose $r_0=10^{-4}$, to achieve good accuracy in \eqref{A3} both in the region $|z|< r_0$ (note that the integral of $1/|z|$ over this region is $\sim r_0^2$) and outside it. 

Finally, let us discuss evaluation of the 1D integrals \eqref{CandSintegrals}, which requires a moment's thought as one needs to resolve both the oscillatory trigonometric factor and the Gaussian factor. The wavevector $\pi p$ is $\le \pi \cdot 2 R_{N}^{\textnormal{Dir}}$ and hence 
$\lesssim 200$
%only grows to about $200$ 
for up to $N=30~000$ electrons (in which case $R_N^{\textnormal{Dir}}\approx 31$), so the trigonometric oscillations can be accurately resolved by any standard quadrature method. The Gaussian factor, however, turns out to be more delicate, as the maximum value of $\alpha_j$ is $\approx 8\times 10^8$. For $\alpha_j>1$ we therefore used  the following alternative expressions obtained by re-scaling:
\begin{align} C_{p,j} & = \frac{1}{\sqrt{\alpha_j}}\int_0^{\sqrt{\alpha_j}} \biggr(\frac{1}{2}-\frac{t}{2\sqrt{\alpha_j}}\biggr) \cos(\pi p \frac{t}{\sqrt{\alpha_j}}) e^{-t^2} \mathrm{d}t,  \label{Cintegralrescaled} \\
S_{p,j} & = \frac{1}{\sqrt{\alpha_j}} \int_0^{\sqrt{\alpha_j}} \frac{1}{2\pi}  \sin(\pi p \frac{t}{\sqrt{\alpha_j}})e^{-t^2} \mathrm{d}t. \label{Sintegralrescaled} \end{align}
Note that even though the integration interval may be big, for practical purposes one can truncate at $\min\{\sqrt{\alpha_j},10\}$ (as $\int_{10}^\infty e^{-t^2} \mathrm{d}t \approx 10^{-45}$). 

\section{Assumptions on GGAs}\label{sec:GGAcheck}
We now show that the expressions for the PBE and B88 functionals (see equations \eqref{pbe} and \eqref{b88}) satisfy our assumptions in Theorem \ref{mainthm}. 

The $C^1$ regularity in $(0,\infty)\times\R$ is straightforward, so we just need to worry about the continuity when $a$ goes to zero and $b$ remains bounded. For this, let us rewrite equations \eqref{b88} and \eqref{pbe} as
\begin{align*}
    &g^{PBE}(a,b) = c_x \frac{\mu s^2}{4(3\pi^2)^\frac23 + \frac{\mu}{\kappa} s^2} a^\frac43 \\
    &g^{B88}(a,b) = \frac{2^\frac13 s^2}{1 + 6 \beta 2^\frac13 s \sinh^{-1}(2^\frac13 s)} a^{\frac43} = \frac{2^\frac13 s}{1 + 6 \beta 2^\frac13 s \sinh^{-1}(2^\frac13 s)} b
\end{align*}
where $s = b/a^\frac43$. Thus for the PBE, since the (enhancement) factor in front of $a^{\frac43}$ is bounded, we see that $g^{PBE}(a,b)\ra 0$ as $a\ra 0$ regardless of $b$. For the B88, we make two observations: (i) thanks to the superlinear growth of the denominator in the enhancement factor, we see that $g^{B88}(a,b)$ goes to zero if $s\ra \infty$ and $b$ stays bounded, and (ii) if $s$ is bounded and $b \ra 0$, $g^{B88}(a,b)$ also goes to zero. In particular, as taking the limit $a \ra 0$ with $b$ bounded falls into one of these two cases, property \eqref{ggaassump} holds.

\begin{small}

\end{small}

\end{document}